\documentclass{article}

\usepackage[a4paper,left=25mm,top=20mm,right=20mm,bottom=15mm]{geometry}

\usepackage{amsmath, amsfonts, amssymb, amsthm}
\usepackage{algorithmic,algorithm}

\providecommand{\keywords}[1]{\textbf{\textit{Keywords: }} #1}

\usepackage{graphicx}
\usepackage{mathptmx} 
\usepackage{color}
\usepackage{algorithmic,algorithm}

\usepackage{mathtools}
\usepackage{natbib}
\setcitestyle{authoryear,open={(},close={)}}
\let\cite\citep
\usepackage{caption}

\usepackage{epstopdf}

\usepackage{subfig}
\usepackage{authblk}

\newcommand{\rev}[1]{\begingroup#1\endgroup}

\newcommand{\Ro}{\mathrel{\overset{0}{\sim}}}
\newcommand{\Rl}{\mathrel{\overset{1}{\sim}}}
\newcommand{\Rlstar}{\mathrel{\overset{1\star}{\sim}}}
\newcommand{\Rldstar}{\mathrel{\overset{1\star}{\rightharpoonup}}}

\newcommand{\Rk}{\mathrel{\overset{\ge k}{\sim}}}
\newcommand{\Rld}{\mathrel{\overset{1}{\rightharpoonup}}}
\newcommand{\Rldk}{\mathrel{\overset{k}{\rightharpoonup}}}
\newcommand{\Rd}{\mathrel{\rightsquigarrow}}
\renewcommand{\P}{\mathbb{P}}
\newcommand{\lca}[1]{\mathop{lca}(#1)}
\newcommand{\rt}[1]{\ensuremath{\mathsf{#1}}}

\newtheorem{thm}{Theorem}
\newtheorem{lemma}[thm]{Lemma}
\newtheorem{cor}[thm]{Corollary}

\newtheorem{definition}[thm]{Definition}

\makeindex             

\begin{document}

\title{Inference of Phylogenetic Trees from the Knowledge of Rare Evolutionary
  		Events}

\author[1,2]{Marc Hellmuth}
\author[3]{Maribel Hernandez-Rosales}
\author[1,8]{Yangjing Long}
\author[4,5,6,7]{Peter F.\ Stadler}

\affil[1]{\footnotesize Dpt.\ of Mathematics and Computer Science, University of Greifswald, Walther-
  Rathenau-Strasse 47, D-17487 Greifswald, Germany \\
	\texttt{mhellmuth@mailbox.org} 	}
\affil[2]{Saarland University, Center for Bioinformatics, Building E 2.1, P.O.\ Box 151150, D-66041 Saarbr{\"u}cken, Germany }
\affil[3]{	CONACYT-Instituto de Matemáticas,		UNAM Juriquilla, 		Blvd. Juriquilla 3001, 76230 Juriquilla, Querétaro, Mexico }
\affil[4]{Bioinformatics Group, Department of Computer Science; 
        Interdisciplinary Center of Bioinformatics; 
        German Centre for Integrative Biodiversity Research (iDiv)
        Halle-Jena-Leipzig; Competence Center for Scalable Data Services
        and Solutions; and Leipzig Research Center for Civilization Diseases,
        Leipzig University,   H{\"a}rtelstra{\ss}e 16-18, D-04107 Leipzig; Germany}
\affil[5]{Max-Planck-Institute for Mathematics in the Sciences, 
  Inselstra{\ss}e 22, D-04103 Leipzig}
\affil[6]{Inst.\ f.\ Theoretical Chemistry, University of Vienna, 
  W{\"a}hringerstra{\ss}e 17, A-1090 Wien, Austria}
\affil[7]{Santa Fe Institute, 1399 Hyde Park Rd., Santa Fe, USA} 
\affil[8]{        School of Mathematics,
        Shanghai Jiao Tong university, Dongchuan Road 800, 200240 Shanghai, 
        China }

\date{}
\normalsize

\maketitle

\abstract{ 
  Rare events have played an increasing role in molecular phylogenetics as
 potentially homoplasy-poor characters. In this contribution we analyze the
 phylogenetic information content from a combinatorial point of view by
 considering the binary relation on the set of taxa defined by the existence of
 a single event separating two taxa. We show that the graph-representation of
 this relation must be a tree. Moreover, we characterize completely the
 relationship between the tree of such relations and the underlying phylogenetic
 tree. With directed operations such as tandem-duplication-random-loss events in
 mind we demonstrate how non-symmetric information constrains the position of
 the root in the partially reconstructed phylogeny }

\smallskip
\noindent

\sloppy

\section{Introduction}

Shared derived characters (synapomorphies or ``Hennigian markers'') that
are unique to specific clades form the basis of classical cladistics
\cite{Hennig:50}. In the context of molecular phylogenetics \emph{rare
  genomic changes (RGCs)} can play the same important role
\cite{Rokas:00,Boore:06}. RGCs correspond to rare mutational events that
are very unlikely to occur multiple times and thus are (almost) free of
homoplasy. A wide variety of processes and associated markers have been
proposed and investigated. Well-studied RGCs include presence/absence
patterns of protein-coding genes \cite{Dutilh:08} as well as microRNAs
\cite{Sempere:06}, retroposon integrations \cite{Shedlock:00}, insertions
and deletions (indels) of introns \cite{Rogozin:05}, pairs of mutually
exclusive introns (NIPs) \cite{Krauss:08a}, protein domains
\cite{Deeds:05,Yang:05}, RNA secondary structures \cite{Misof:03}, protein
fusions \cite{Stechmann:03}, changes in gene order
\cite{Sankoff:82,Boore:98,Lavrov:07}, metabolic networks
\cite{Forst:01,Forst:06a,Mazurie:08}, transcription factor binding sites
\cite{Prohaska:04a}, insertions and deletions of arbitrary sequences
\cite{Simmons:00,Ashkenazy:14,Donath:14a}, and variations of the genetic
code \cite{Abascal:12}. RGCs clearly have proved to be phylogenetically
informative and helped to resolve many of the phylogenetic questions where
sequence data lead to conflicting or equivocal results.

Not all RGCs behave like cladistic characters, however. While
presence/absence characters are naturally stored in character matrices
whose columns can vary independently, this is not the case e.g.\ for gene
order characters. From a mathematical point of view, character-based
parsimony analysis requires that the mutations have a product structure in
which characters are identified with factors and character states can vary
independently of each other \cite{Wagner:03a}. This assumption is violated
whenever changes in the states of two distinct characters do not commute.
Gene order is, of course, the prime example on non-commutative events.
  
Three strategies have been pursued in such cases\rev{: (i)} Most
importantly, the analog of the parsimony approach is considered for a
particular non-commutative model.  For the genome rearrangements an
elaborated theory has been developed that considers various types of
operations on (usually signed) permutations. Already the computation of
editing distances is non-trivial. An added difficulty is that the interplay
of different operations such as reversals, transpositions, and
tandem-duplication-random-loss (TDRL) events is difficult to handle
\cite{Bernt:07a,Hartmann:16}. \rev{(ii)} An alternative is to focus on
distance-based methods \cite{Wang:06}. Since good rate models are usually
unavailable, distance measures usually are not additive and thus fail to
precisely satisfy the assumptions underlying the most widely used methods
such as neighbor joining. \rev{(iii) Finally, the non-commutative data
  structure can be converted into} a presence-absence structure, e.g., by
using pairwise adjacencies \cite{Tang:05} as a representation of
permutations or using list alignments in which rearrangements appear as
pairs of insertions and deletions \cite{Fritzsch:06a}. While this yields
character matrices that can be fed into parsimony algorithms, these can
only result in approximate heuristics.

While it tends to be difficult to disentangle multiple, super-imposed
complex changes such as genome rearrangements or tandem duplication, it is
considerably simpler to recognize whether two genes or genomes differ by a
single RGC operation. It make sense therefore to ask just how much
phylogenetic information can be extracted from elementary RGC events. Of
course, we cannot expect that a single RGC will allow us to (re)construct a
detailed phylogeny. It can, however, provide us with solid, well-founded
constraints. Furthermore, we can hope that the combination of such
constraints can be utilized as a practicable method for phylogenetic
inference. Recently, we have shown that orthology assignments in large gene
families imply triples that must be displayed by the underlying species
tree \cite{HernandezRosales:12a,Hellmuth:13a}. In a phylogenomics setting a
sufficient number of such triple constraints can be collected to yield
fully resolved phylogenetic trees \cite{Hellmuth:15a}, see \citet{HW:16b}
for an overview.

A plausible application scenario for our setting is the rearrangement of
mitogenomes \cite{Sankoff:82}. Since mitogenomes are readily and cheaply
available, the taxon sampling is sufficiently dense so that the gene orders
often differ by only a single rearrangement or not at all. These cases are
identifiable with near certainty \cite{Bernt:07a}. Moreover, some RGC are
inherently directional. Probably the best known example is the tandem
duplication random loss (TDRL) operation \cite{Chaudhuri:06}. We will
therefore also consider a directed variant of the problem.

In this contribution, we ask how much phylogenetic information can be
retrieved from single RGCs. More precisely, we consider a scenario in which
we can, for every pair of taxa distinguish, for a given type of RGC,
whether $x$ and $y$ have the same genomic state, whether $x$ and $y$ differ
by exactly one elementary change, or whether their distance is larger than
a single operation. We formalize this problem in the following way. Given a
relation $\sim$, there is a phylogenetic tree $T$ with an edge labeling
$\lambda$ (marking the elementary events) such that $x\sim y$ if and only
if the edge labeling along the unique path $\mathbb{P}(x,y)$ from $x$ to $y$ in
$T$ has a certain prescribed property $\Pi$. After defining the necessary
notation and preliminaries, we give a more formal definition of the general
problem in section~\ref{sect:theory}.

\rev{The graphs defined by path relations on a tree are closely related to
  \emph{pairwise compatibility graphs} (PCGs). A graph $G =(V,E)$ is a PCG
  if there is a tree $T$ with leaf set $V$, a positive edge-weight function
  $w:E(T)\to \mathbb{R}^+$, and two nonnegative real numbers $d_{\min}\le
  d_{\max}$ such that there is an edge $uv \in E(G)$ if and only if
  $d_{\min}\leq d_{T,w}(x,y) \leq d_{\max}$, where $d_{T,w}(x,y)$ is the
  sum of the weights of the edges on the unique path $\mathbb{P}(x,y)$ in
  $T$. One writes $G = \mathrm{PCG}(T, w, d_{\min} , d_{\max})$.} In this
contribution we will primarily be interested in the special case where
$\Pi$ is ``a single event along the path''. \rev{Although PCGs have been
  studied extensively, see e.g.,
  \citet{PCGsurvey,YHTR:08,YBR:10,CMPS:13,MR:13,DMR:13}, the questions are
  different from our motivation and, to our knowledge, no results have been
  obtained that would simplify the characterization of the PCGs
  corresponding to the ``single-1-relation'' in Section~\ref{sect:single1}.
  Furthermore, PCGs are always treated as undirected graphs in the
  literature.  We also consider an antisymmetric (Section~\ref{sect:1dir})
  and a general directed (Section~\ref{sect:mixed}) versions of the
  single-1-relation motivated by RGCs with directional information.}

\rev{The main result of this contribution can be summarized as follows: (i)
  The graph of a single-1-relation is always a forest. (ii) If the
  single-1-relation is connected, there is a unique minimally resolved tree
  that explains the relation. The same holds true for the connected
  components of an arbitrary relation. (iii) Analogous results hold for the
  anti-symmetric and the mixed variants of the single-1-relation. In this
  case not only the tree topology but also the position of the root can be
  determined or at least constrained.}  \rev{Together, these results in a
  sense characterize the phylogenetic information contained in rare events:
  if the single-1-relation graph is connected, it is a tree that through a
  bijection corresponds to a uniquely defined, but not necessarily fully
  resolved, phylogenetic tree. Otherwise, it is forest whose connected
  components determine subtrees for which the rare events provide at least
  some phylogenetically relevant information.}

\section{Preliminaries}
\label{sec:prelim}

\subsection{Basic Notation} 

We largely follow the notation and terminology of \rev{the book by}
\citet{sem-ste-03a}. Throughout, $X$ denotes always a finite set of at
least three taxa. We will consider both undirected and directed graphs
$G=(V,E)$ with finite vertex set $V(G)\coloneqq V$ and edge set or arc set
$E(G)\coloneqq E$.  For a digraph $G$ we write $\underline{G}$ for its
\emph{underlying undirected graph} where $V(G)=V(\underline{G})$ and
$\{x,y\}\in E(\underline{G})$ if $(x,y)\in E(G)$ or $(y,x)\in E(G)$. Thus,
$\underline{G}$ is obtained from $G$ by ignoring the direction of edges.
For simplicity, edges $\{x,y\}\in E(G)$ (in the undirected case) and arcs
$(x,y)\in E(G)$ (in the directed case) will be both denoted by $xy$.

The representation $G(R)=(V,E)$ of a relation $R\subseteq V\times V$ has
vertex set $V$ and edge set $E=\{xy\mid (x,y)\in R\}$. If $R$ is
irreflexive, then $G$ has no loops. If $R$ is symmetric, we regard $G(R)$
as an undirected graph. A \emph{clique} is a complete subgraph that is
maximal w.r.t.\ inclusion. An equivalence relation is \emph{discrete} if
all its equivalence classes consist of single vertices.

A tree $T=(V,E)$ is a connected cycle-free undirected graph.  The vertices
of degree $1$ in a tree are called leaves, all other vertices of $T$ are
called \emph{inner vertices}. 
An edge of $T$ is \emph{interior} if both of
its end vertices are inner vertices, otherwise the edge is called \emph{terminal}.
\rev{For technical reasons, we call a vertex $v$ an inner vertex and leaf
if $T$ is a single vertex graph
$(\{v\},\emptyset)$. However, if 
$T$ is an edge $vw$ we refer to $v$ and $w$ as leaves but not
as inner vertices. Hence, in this case the edge $vw$ is not an
interior edge}

A \emph{star} $S_m$ with $m$ leaves
is a tree that has at most one inner vertex. A \emph{path} $P_n$ (on $n$
vertices) is a tree with two leaves and $n-3$ interior edges.  There is a
unique path $\mathbb{P}(x,y)$ connecting any two vertices $x$ and $y$ in a
tree $T$. We write $e\in\mathbb{P}(x,y)$ if the edge $e$ connects two
adjacent vertices along $\mathbb{P}(x,y)$.  We say that a directed graph is
a tree if its underlying undirected graph is a tree.  A directed path $P$
is a tree on vertices $x_1,\dots,x_n$ s.t. $x_ix_{i+1}\in E(P)$, $1\leq
i\leq n-1$. A graph is a forest if all its connected components are trees.

A tree is \emph{rooted} if there is a distinguished vertex $\rho\in V$
called the \emph{root}. Throughout this contribution we assume that the
root is an inner vertex. \rev{Given a rooted tree $T=(V,E)$, there is a
  partial order $\preceq$ on $V$ defined as $ v \preceq u$ if $u$ lies on
  the path from $v$ to the root. Obviously, the root is the unique maximal
  element w.r.t\ $\preceq$.  For a non-empty subset of $W\subseteq V$, we
  define $\lca{W}$, or the \emph{least common ancestor of $W$}, to be the
  unique $\preceq_T$-minimal vertex of $T$ that is an ancestor of every
  vertex in $W$. In case $W=\{x,y \}$, we put $\lca{x,y}:=\lca{\{x,y\}}$.}
If $T$ is rooted, then by definition $\lca{x,y}$ is a uniquely defined
inner vertex along $\mathbb{P}(x,y)$.

We write $L(v)$ for the set of leaves in the subtree below a fixed vertex
$v$, i.e., $L(v)$ is the set of all leaves for which $v$ is located on the
unique path from $x\in L(v)$ to the root of $T$. The \emph{children} of an
inner vertex $v$ are its direct descendants, i.e., vertices $w$ with $vw\in
E(T)$ s.t.\ that $w$ is further away from the root than $v$.  A rooted or
unrooted tree that has no vertices of degree two \rev{(except possibly the
  root of $T$)} and leaf set $X$ is called a \emph{phylogenetic tree $T$
  (on $X$)}.

Suppose $X'\subseteq X$. A phylogenetic tree $T$ on $X$ \emph{displays} a
phylogenetic tree $T'$ on $X'$ if $T'$ can be obtained from $T$ by a series
of vertex deletions, edge deletions, and suppression of vertices of degree
$2$ other than possibly the root, i.e., the replacement of an inner
vertex $u$ and its two incident edges $e'$ and $e''$ by a single edge $e$,
cf.\ \rev{Def.\ 6.1.2 in the book by \citet{sem-ste-03a}}. In the rooted
case, only a vertex between two \rev{incident} edges may be suppressed;
furthermore, if $X'$ is contained in a single subtree, then the $\lca{X'}$
becomes the root of $T'$. It is useful to note that $T'$ is displayed by
$T$ if and only if it can be obtained from $T$ step-wisely by removing an
arbitrarily selected leaf $y\in X\setminus X'$, its incident edge $e=yu$,
and suppression of $u$ provided $u$ has degree $2$ after removal of $e$.

We say that a rooted tree $T$ \emph{contains} or \emph{displays} the triple
$\rt{xy|z}$ if $x,y,$ and $z$ are leaves of $T$ and the path from $x$ to
$y$ does not intersect the path from $z$ to the root of $T$. A set of
triples $\mathcal R$ is consistent if there is a rooted tree that contains
all triples in $\mathcal R$.  For a given leaf set $L$, a triple set
$\mathcal R$ is said to be \emph{strict dense} if for any three distinct
vertices $x,y,z\in L$ we have $|\{\rt{xy|z}, \rt{xz|y}, \rt{yz|x}\}\cap
R|=1$. It is well-known that any consistent strict-dense triple set
$\mathcal R$ has a unique representation as a binary tree
\cite[Suppl. Material]{Hellmuth:15a}.  For a consistent set $R$ of rooted
triples we write $R\vdash \rt{(xy|z)}$ if any phylogenetic tree that
displays all triples of $R$ also displays $\rt{(xy|z)}$.
\rev{\citet{BS:95} extend and generalized results by \citet{Dekker86} and
  showed} under which conditions it is possible to infer triples by using
only subsets $R'\subseteq R$, i.e., under which conditions $R\vdash
\rt{(xy|z)} \implies R'\vdash \rt{(xy|z)}$ for some $R'\subseteq R$. In
particular, we will use the following inference rules:
\renewcommand{\theequation}{\roman{equation}}
\begin{align}
  \{\rt{(ab|c)}, \rt{(ad|c)}\} &\vdash \rt{(bd|c)} 
  \label{eq:infRule1} \\
  \{\rt{(ab|c)}, \rt{(ad|b)}\} & \vdash \rt{(bd|c)},\rt{(ad|c)} 
  \label{eq:infRule2} \\
  \{\rt{(ab|c)}, \rt{(cd|b)}\} &\vdash \rt{(ab|d)},\rt{(cd|a)}.
  \label{eq:infRule3}
\end{align}

\section{Path Relations and Phylogenetic Trees}
\label{sect:theory} 

Let $\Lambda$ be a non-empty set. Throughout this contribution we consider
a \rev{phylogenetic tree} $T=(V,E)$ with edge-labeling $\lambda \colon E\to
\Lambda$. An edge $e$ with label $\lambda(e)=k$ will be called a
\emph{k-edge}. We interpret $(T,\lambda)$ so that a RGC occurs along edge
$e$ if and only if $\lambda(e)=1$. Let $\Pi$ be a subset of the set of
$\Lambda$-labeled paths. We interpret $\Pi$ as a property of the path and
its labeling. The tree $(T,\lambda)$ and the property $\Pi$ together define
a binary relation $\sim_{\Pi}$ on $X$ by setting
\begin{equation}
  x\sim_{\Pi} y \quad\iff\quad (\mathbb{P}(x,y),\lambda) \in \Pi 
\end{equation} 
The relation $\sim_{\Pi}$ has the graph representation $G(\sim_{\Pi})$ with
vertex set $X$ and edges $xy\in E(G(\sim_{\Pi}))$ if and only if
$x\sim_{\Pi} y$.

\begin{definition}
  Let $(T,\lambda)$ be a $\Lambda$-labeled phylogenetic tree with leaf set
  $L(T)$ and let $G$ be a graph with vertex set $L(T)$. We say that
  \emph{$(T,\lambda)$ explains $G$ (w.r.t.\ to the path property $\Pi$)} if
  $G=G(\sim_{\Pi})$.
\end{definition}
For simplicity we also say ``$(T,\lambda)$ explains $\sim$'' for the binary
relation $\sim$.

We consider in this contribution the conceptually ``inverse problem'':
Given a definition of the predicate $\Pi$ as a function of edge labels
along a path and a graph $G$, is there an edge-labeled tree $(T,\lambda)$
that explains $G$? Furthermore, we ask for a characterization of the class
of graph that can be explained by edge-labeled trees and a given predicate
$\Pi$.

A straightforward biological interpretation of an edge labeling
\rev{$\lambda: E\to \{0,1\}$} is that a certain type of evolutionary event
has occurred along $e$ if and only if $\lambda(e)=1$. This suggests that in
particular the following path properties and their associated relations on
$X$ are of practical interest:
\begin{itemize}
\item[{$x \Ro y$}] if and only if all edges in $\mathbb{P}(x,y)$ are
  labeled $0$; For convenience we set $x \Ro x$ for all $x\in X$.
\item[{$x \Rl y$}] if and only if all but one edges along $\mathbb{P}(x,y)$
  are labeled $0$ and exactly one edge is labeled $1$;
\item[{$x \Rld y$}] if and only if all edges along $\mathbb{P}(u,x)$ are 
  labeled $0$ and exactly one edge along $\mathbb{P}(u,y)$ is labeled $1$,
  where $u=\lca{x,y}$. 
\item[{$x \Rk y$}] with \rev{$k\geq1$} if and only if at least $k$
  edges \rev{along $\mathbb{P}(x,y)$} are labeled $1$;
\item[{$x \Rd y$}] \rev{if all edges along $\mathbb{P}(u,x)$ are labeled
    $0$ and there are one or more edges along $\mathbb{P}(u,y)$ with a
    non-zero label, where $u=\lca{x,y}$.}
\end{itemize}
We will call the relation $\Rl$ the \emph{single-1-relation}. It will be
studied in detail in the following section. Its directed variant $\Rld$
will be investigated in Section~\ref{sect:1dir}. The more general relations
$\Rk$ and $\Rd$ will be studied \rev{in future work.} 

\rev{As noted in the introduction there is close relationship between the
  graphs of path relations introduced above and PCGs. For instance, the
  single-1-relations correspond to a graph of the form $G =
  \mathrm{PCG}(T,\lambda,1,1)$ for some tree $T$. The exact-$k$ leaf power
  graph $\mathrm{PCG}(T,\lambda,k,k)$ arise when $\lambda(e)=1$ for all
  $e\in E(T)$ \cite{BVR:10}. The ``weight function'' $\lambda$, however,
  may be $0$ in our setting. It is not difficult to transform our weight
  functions to strictly positive values albeit at the expense of using less
  ``beautiful'' values of $d_{\min}$ and $d_{\max}$. The literature on the
  PCG, to our knowledge, does not provide results that would simplify our
  discussion below. Furthermore, the applications that we have in mind for
  future work are more naturally phrased in terms of Boolean labels, such as
  the ``at least one 1'' relation, or even vector-valued structures. We
  therefore do not pursue the relationship with PCGs further.}

The combinations of labeling systems and path properties of primary
interest to us have \emph{nice properties}:
\begin{itemize}
\item[(L1)] The label set $\Lambda$ is endowed with a semigroup $\boxplus:
  \Lambda\times\Lambda\to\Lambda$.
\item[(L2)] There is a subset $\Lambda_{\Pi}\subseteq\Lambda$ of labels
  such that $(\mathbb{P}(x,y),\lambda) \in \Pi$ if and only if
  $\lambda(\mathbb{P}(x,y)):=\boxplus_{e\in\mathbb{P}(x,y)}\lambda(e) \in
  \rev{\Lambda_{\Pi}}$ or
  $\lambda(\mathbb{P}(x,y)):=\boxplus_{e\in\mathbb{P}(\lca{x,y},y)}\lambda(e)
  \rev{\in \Lambda_{\Pi}}$.
\end{itemize}
For instance, we may set $\Lambda=\mathbb{N}$ and use the usual addition
for $\boxplus$. Then $\Ro$ corresponds to $\Lambda_{\Pi}=\{0\}$, $\Rl$
corresponds to $\Lambda_{\Pi}=\{1\}$, etc. \rev{The bounds $d_{\min}$ and
  $d_{\max}$ in the definition of PCGs of course is just a special case of
  of the predicate $\Lambda_{\Pi}$.}

We now extend the concept of a \rev{phylogenetic tree} displaying another
one to the $\Lambda$-labeled case.
\begin{definition}
  Let $(T,\lambda)$ and $(T',\lambda')$ be two \rev{phylogenetic trees}
  with $L(T')\subseteq L(T)$. Then $(T,\lambda)$ displays $(T',\lambda')$
  \rev{w.r.t.\ a path property $\Pi$} if (i) $T$ displays $T'$ and (ii)
  $(\mathbb{P}_{T}(x,y),\lambda) \in \rev{\Lambda_{\Pi}}$ if and only if
  $(\mathbb{P}_{T'}(x,y),\rev{\lambda'}) \in \rev{\Lambda_{\Pi}}$ for all
  $x,y\in L(T')$.
\end{definition} 
The definition is designed to ensure that the following property is satisfied:
\begin{lemma}
  Suppose $(T,\lambda)$ displays $(T',\lambda')$ and $(T,\lambda)$ explains
  a graph $G$. Then $(T',\lambda')$ explains the induced subgraph
  $G[L(T')]$.
\end{lemma}

\begin{lemma} 
  Let $(T,\lambda)$ display $(T',\lambda')$.  Assume that the labeling
  system satisfies (L1) and (L2) and suppose
  $\lambda'(e)=\lambda(e')\boxplus\lambda(e'')$ whenever $e$ is the edge
  resulting from suppressing the inner vertex between $e'$ and $e''$.
  If $T'$ is displayed by $T$ then $(T',\lambda')$ is displayed by
  $(T,\lambda)$ \rev{(w.r.t.\ any path property $\Pi$).}
\end{lemma}
\begin{proof} 
  Suppose $T'$ is obtained from $T$ by removing a single leaf $w$. By
  construction $T'$ is displayed by $T$ and $\lambda(\mathbb{P}(x,y))$ is
  preserved upon removal of $w$ and suppression of its neighbor. Thus  
  (L2) implies that $(T,\lambda)$ displays $(T',\lambda')$. For an arbitrary 
  $T'$ displayed by $T$ this argument can be repeated for each individual
  leaf removal on the editing path from $T$ to $T'$.  
\end{proof}
\rev{We note in passing that this construction is also well behaved for
  PCGs: it preserves path length, and thus distances between leaves, by
  summing up the weights of edges whenever a vertex of degree 2 between
  them is omitted.}

Let us now turn to the properties of the specific relations that are of
interest in this contribution.

\begin{lemma}
  The relation $\Ro$ is an equivalence relation.
\end{lemma}
\begin{proof}
  By construction, $\Ro$ is symmetric and reflexive. To establish
  transitivity, suppose $x\Ro y$ and $y\Ro z$, i.e., $\lambda(e)=0$ for all
  $e\in \mathbb{P}(x,y)\cup \mathbb{P}(y,z)$. By uniqueness of the path
  connecting vertices in a tree, $\mathbb{P}(x,z)\subseteq
  \mathbb{P}(x,y)\cup \mathbb{P}(y,z)$, i.e., $\lambda(e)=0$ for all $e\in
  \mathbb{P}(x,z)$ and therefore $x \Ro z$.  
\end{proof}
Since $\Ro$ is an equivalence relation, the graph $G(\Ro)$ is a disjoint
union of complete graphs, or in other words, each connected component of
$G(\Ro)$ is a clique.

We are interested here in characterizing the pairs of trees and labeling
functions $(T,\lambda)$ that explain a given relation $\rho$ as its $\Ro$,
$\Rl$ or $\Rld$ relation. More precisely, we are interested in the least
resolved trees with this property.

\begin{definition} 
  Let $(T,\lambda)$ be an edge-labeled phylogenetic tree with leaf set
  $X=L(T)$. We say that $(T',\lambda')$ is \emph{edge-contracted from
  $(T,\lambda)$} if the following conditions hold: (i) $T'=T/e$ is the
  usual graph-theoretical edge contraction for some interior edge
  $e=\{u,v\}$ of $T$.
  
  \noindent (ii) The labels satisfy $\lambda'(e')=\lambda(e')$ for
  all $e'\ne e$.
\end{definition} 
Note that we do not allow the contraction of terminal edges, i.e., of edges
incident with leaves. 

\rev{
\begin{definition}[Least and Minimally Resolved Trees]
  Let $R\in \{\Ro,\ \Rl,\ \Rld,\ \Rl/\Ro,\ \Rld/\Ro\}$. 
  A pair $(T,\lambda)$ is \emph{least resolved} for a prescribed relation 
  $R$ if no edge contraction leads to a tree
  $T',\lambda')$ of $(T,\lambda)$ that explains $R$.
  A pair $(T,\lambda)$ is \emph{minimally resolved} for a
  prescribed relation $R$ if it has the fewest number of
  vertices among all trees that explain $R$.
\end{definition}

Note that every minimally resolved tree is also least resolved, but not
\textit{vice versa}.  }

\section{The single-1-relation}
\label{sect:single1}

The single-1-relation does not convey any information on the location of
the root and the corresponding partial order on the tree. We therefore
regard $T$ as unrooted in this section.

\begin{lemma}
  \label{lem:notriangle}
  Let $(T,\lambda)$ be an edge-labeled \rev{phylogenetic tree with leaf set
    $X$ and} resulting relations $\Ro$ and $\Rl$ over $X$.  Assume that
  $A,B$ are distinct cliques in $G(\Ro)$ and suppose $x\Rl y$ where $x\in
  A$ and $y\in B$. Then $x'\Rl y'$ holds for all $x'\in A$ and $y'\in B$.
\end{lemma}
\begin{proof}
  First, observe that $\mathbb{P}(x',y') \subseteq
  \mathbb{P}(x',x)\cup \mathbb{P}(x,y)\cup \mathbb{P}(y,y')$ in $T$. 
  Moreover,  $\mathbb{P}(x',x)$ and $\mathbb{P}(y,y')$ have only edges with 
  label $0$. As $\mathbb{P}(x,y)$ contains exactly one non-0-label, thus 
  $\mathbb{P}(x',y')$ contains at most one non-0-label. If there was 
  no non-0-label, then $\mathbb{P}(x,y) \subseteq
  \mathbb{P}(x,x')\cup \mathbb{P}(x',y')\cup \mathbb{P}(y',y)$ would imply 
  that $\mathbb{P}(x,y)$ also has only 0-labels, a contradiction. Therefore 
  $x'\Rl y'$.  
\end{proof}

As a consequence it suffices to study the single-1-relation on the quotient
graph $G(\Rl)/\Ro$. To be more precise, $G(\Rl)/\Ro$ has as vertex set the
equivalence classes of $\Ro$ and two vertices $c_i$ and $c_j$ are connected
by an edge if there are vertices $x\in c_i$ and $y\in c_j$ with $x\Rl
y$. Analogously, the graph $G(\Rld)/\Ro$ is defined.

  For a given $(T,\lambda)$ and its corresponding relation $\Ro$
  consider an arbitrary nontrivial equivalence class $c_i$ of $\Ro$. Since
  $\Ro$ is an equivalence relation, the induced subtree $T'$ with leaf set
  $c_i$ and inner vertices $\lca{c}$ for any subset $c\subseteq c_i$
  contains only 0-edges and is maximal w.r.t. this property. Hence, we
  could remove $T'$ from $T$ and identify the root $\lca{c_i}$ of $T'$ in
  $T$ by a representative of $c_i$, while keeping the information of $\Ro$
  and $\Rl$. Let us be a bit more explicit about this point. Consider trees
  $(T_Y,\lambda_Y)$ displayed by $(T,\lambda)$ with leaf sets $Y$ such that
  $Y$ contains exactly one (arbitrarily chosen) representative from each
  $\Ro$ equivalence class of $(T,\lambda)$.  For any such trees
  $(T_Y,\lambda_Y)$ and $(T_Y',\lambda_Y')$ with the latter property, there
  is an isomorphism $\alpha: T_Y\to T_{Y'}$ such that $\alpha(y)\Ro y$ and
  $\lambda_{Y'}(\alpha(e))=\lambda_{Y'}(e)$. Thus, all such
  $(T_Y,\lambda_Y)$ are isomorphic and differ basically only in the choice
  of the particular representatives of the equivalence classes of $\Ro$.
  Furthermore, $T_{Y}$ is isomorphic to the quotient graph $T/\Ro$ obtained
  from $T$ by replacing the (maximal) subtrees where all edges are labeled
  with $0$ by a representative of the corresponding $\Ro$-class.  Suppose
  $(T,\lambda)$ explains $G$. Then $(T_Y,\lambda_Y)$ explains $G[Y]$ for a
  given $Y$. Since all $(T_Y,\lambda_Y)$ are isomorphic, all $G[Y]$ are
  also isomorphic, and thus $G[Y]=G/\Ro$ for all $Y$.
  
To avoid unnecessarily clumsy language we will say that ``$(T,\lambda)$
explains \hbox{$G(\Rl)/\Ro$}'' instead of the more accurate wording
``$(T,\lambda)$ displays $(T_Y,\lambda_Y)$ where $Y$ contains exactly one
representative of each $\Ro$ equivalence class such that $(T_Y,\lambda_Y)$
explains $G(\Rl)/\Ro$''.

In contrast to $\Ro$, the single-1-relation $\Rl$ is not transitive.  As an
example, consider the star $S_3$ with leaf set $\{x,y,z\}$, inner
vertex $v$, and edge labeling $\lambda(v,x)=\lambda(v,z)=1\neq
\lambda(v,y)=0$.  Hence $x\Rl y$, $y\Rl z$ and $x\not \Rl z$. In fact, a
stronger property holds that forms the basis for understanding the
single-1-relation:
\begin{lemma}
  If $x\Rl y$ and $x\Rl z$, then $y\not\Rl z$.
	\label{lem:not-trans}
\end{lemma}
\begin{proof}
  Uniqueness of paths in $T$ implies that there is a unique inner vertex
  $u$ in $T$ such that
  $\mathbb{P}(x,y)=\mathbb{P}(x,u)\cup\mathbb{P}(u,y)$,
  $\mathbb{P}(x,z)=\mathbb{P}(x,u)\cup\mathbb{P}(u,z)$,
  $\mathbb{P}(y,z)=\mathbb{P}(y,u)\cup\mathbb{P}(u,z)$.  By assumption,
  each of the three sub-paths $\mathbb{P}(x,u)$, $\mathbb{P}(y,u)$, and
  $\mathbb{P}(z,u)$ contains at most one 1-label. There are only two cases:
  (i) There is a 1-edge in $\mathbb{P}(x,u)$. Then neither
  $\mathbb{P}(y,u)$ nor $\mathbb{P}(z,u)$ may have another 1-edge, and thus
  $y\Ro z$, which implies that $y\not\Rl z$. (ii) There is no 1-edge in
  $\mathbb{P}(x,u)$. Then both $\mathbb{P}(y,u)$ and $\mathbb{P}(z,u)$ must
  have exactly one 1-edge. Thus $\mathbb{P}(y,z)$ harbors exactly two
  1-edges, whence $y\not\Rl z$.   \end{proof}

Lemma \ref{lem:not-trans} can be generalized as follows.
\begin{lemma}
  Let $x_1,\dots,x_n$ be vertices s.t.\ $x_{i}\Rl x_{i+1}$, $1\leq i\leq
  n-1$.  Then, for all $i,j$, $x_i\Rl x_j$ if and only if $|i-j|=1$.
  \label{lem:cycle-free}
\end{lemma}
\begin{proof}
  For $n=3$, we can apply Lemma \ref{lem:not-trans}. Assume the assumption
  is true for all $n<K$. Now let $n=K$. Hence, for all vertices $x_i,x_j$
  along the paths from $x_1$ to $x_{K-1}$, as well as the paths from $x_2$
  to $x_K$ it holds that $|i-j|=1$ if and only if we have $x_i\Rl
  x_j$. Thus, for the vertices $x_i,x_j$ we have $|i-j|>1$ if and only if
  we have $x_i\not \Rl x_j$. Therefore, it remains to show that $x_1\not\Rl
  x_n$.
	
  Assume for contradiction, that $x_1\Rl x_n$.  Uniqueness of paths on $T$
  implies that there is a unique inner vertex $u$ in $T$ that lies on
  all three paths $\mathbb{P}(x_1,x_2)$, $\mathbb{P}(x_1,x_n)$, and
  $\mathbb{P}(x_2,x_n)$.

  There are two cases, either there is a 1-edge in $\mathbb{P}(x_1,u)$ or
  $\mathbb{P}(x_1,u)$ contains only 0-edges.

  If $\mathbb{P}(x_1,u)$ contains a 1-edge, then all edges along the path
  $\mathbb{P}(u, x_n)$ must be $0$, and all the edge on path $\mathbb{P}(u,
  x_2)$ must be $0$, However, this implies that $x_2\Ro x_n$, a
  contradiction, as we assumed that $\Ro$ is discrete.

  Thus, there is no 1-edge in $\mathbb{P}(x_1,u)$ and hence, both paths
  $\mathbb{P}(u,x_n)$ and $\mathbb{P}(u,x_2)$ contain each exactly one
  1-edge.
		
  Now consider the unique vertex $v$ that lies on all three paths
  $\mathbb{P}(x_1,x_2)$, $\mathbb{P}(x_1,x_3)$, and $\mathbb{P}(x_2,x_3)$.

  Since $u,v\in \mathbb{P}(x_1,x_2)$, we have either (A) $v\in
  \mathbb{P}(x_1,u)$ where $u=v$ is possible, or (B) $u\in
  \mathbb{P}(x_1,v)$ and $u\neq v$. We consider the two cases separately.

  \smallskip Case (A): Since there is no 1-edge in $\mathbb{P}(x_1,u)$ and
  $x_1\Rl x_n$, resp., $x_1\Rl x_2$ there is exactly one 1-edge in
  $\mathbb{P}(u,x_n)$, resp., $\mathbb{P}(u,x_2)$.  Moreover, since $x_2\Rl
  x_3$ the path $\mathbb{P}(v,x_3)$ contains only 0-edges, and thus $x_3\Rl
  x_n$, a contradiction.

  \smallskip Case (B): Since there is no 1-edge in
  $\mathbb{P}(x_1,u)$ and $x_1\Rl x_n$, the path $\mathbb{P}(u,x_n)$
  contains exactly one 1-edge.

  In the following, we consider paths between two vertices $x_i, x_{n-i}
  \in \{x_1,\dots,x_n\}$ step-by-step, starting with $x_1$ and $x_{n-1}$.

  The induction hypothesis implies that $x_1\not \Rl x_{n-1}$ and since $\Ro$
  is discrete, we can conclude that $x_1\not \Ro x_{n-1}$.  Let
  $\mathbb{P}(x_1,x_n) = \mathbb{P}(x_1,a)\cup ab \cup \mathbb{P}(b,x_n)$
  where $e=ab$ is the 1-edge contained in $\mathbb{P}(x_1,x_n)$.  Let $c_1$
  be the unique vertex that lies on all three paths $\mathbb{P}(x_1,x_n),
  \mathbb{P}(x_1,x_{n-1})$, and $\mathbb{P}(x_{n-1},x_n)$.  If $c_1$ lies
  on the path $\mathbb{P}(x_1,a)$, then $\mathbb{P}(c_1,x_{n-1})$ contains
  only 0-edges, since $\mathbb{P}(x_{n-1},x_n) =
  \mathbb{P}(x_{n-1},c_1)\cup \mathbb{P}(c_1,a)\cup ab \cup
  \mathbb{P}(b,x_n)$ and $x_n \Rl x_{n-1}$. However, in this case the path
  $\mathbb{P}(x_{1},c_1)\cup\mathbb{P}(c_1, x_{n-1})$ contains only
  0-edges, which implies that $x_1\Ro x_{n-1}$, a contradiction.  Thus, the
  vertex $c_1$ must be contained in $\mathbb{P}(b,x_n)$.  Since $x_1\Rl
  x_n$, the path $\mathbb{P}(c_1,x_n)$ contains only 0-edges.  Hence, the
  path $\mathbb{P}(c_1,x_{n-1})$ contains exactly one 1-edge, because $x_n
  \Rl x_{n-1}$. In particular, by construction we see that
  $\mathbb{P}(x_1,x_n) = \mathbb{P}(x_1,u)\cup
  \mathbb{P}(u,c_1)\cup\mathbb{P}(c_1,x_n)$.

  Now consider the vertices $x_n$ and $x_{n-2}$. Let $a'b'$ be the 1-edge
  on the path $\mathbb{P}(c_1,x_{n-1}) = \mathbb{P}(c_1,a') \cup a'b' \cup
  \mathbb{P}(b',x_{n-1})$.  Since $x_{n}\not \Rl x_{n-2}$ and $x_{n}\not
  \Ro x_{n-2}$ we can apply the same argument and conclude that there is a
  vertex $c_2\in \mathbb{P}(b',x_{n-1})$ s.t.\ the path
  $\mathbb{P}(c_2,x_{n-2})$ contains exactly one 1-edge.  In particular, by
  construction we see that $\mathbb{P}(x_1,x_{n-2}) = \mathbb{P}(x_1,u)\cup
  \mathbb{P}(u,c_1)\cup \mathbb{P}(c_1,c_2) \cup\mathbb{P}(c_2,x_{2})$
  s.t.\ the path $\mathbb{P}(c_1,c_2)$ contains exactly one 1-edge.

  Repeating this argument, we arrive at vertices $x_2$ and $x_4$ and can
  conclude analogously that there is a path $\mathbb{P}(c_{n-2},x_2)$ that
  contains exactly one 1-edge and in particular, that $\mathbb{P}(x_1,x_2)
  = \mathbb{P}(x_1,u)\cup \mathbb{P}(u,c_1)\cup \left(\bigcup_{1\leq i\leq
      n-3}\mathbb{P}(c_i,c_{i+1})\right) \cup\mathbb{P}(c_{n-2},x_{n-2})$,
  where each of the distinct paths $\mathbb{P}(c_i,c_{i+1})$, $1\leq i\leq
  n-3$ contains exactly one 1-edge.  However, this contradicts that $x_1\Rl
  x_2$.   
\end{proof}

\begin{cor}
  The graph $G(\Rl)/\Ro$ is a forest, and hence all paths in $G(\Rl)/\Ro$
  are induced paths.
  \label{cor:cycle-free}
\end{cor}

Next we analyze the effect of edge contractions in $T$.

\begin{lemma} \label{lem:contract}
  Let $(T,\lambda)$ explain $G(\Rl)/\Ro$ and let $(T',\lambda')$ be the
  result of contracting an \rev{interior} edge  $e$ in $T$. If $\lambda(e)=0$
  then $(T',\lambda')$ explains $G(\Rl)/\Ro$. If $G(\Rl)/\Ro$ is connected
  and $\lambda(e)=1$ then $(T',\lambda')$ does not explain $G(\Rl)/\Ro$.

  \rev{If $G(\Rl)/\Ro$ is connected and $(T,\lambda)$ is a tree that
    explains \mbox{$G(\Rl)/\Ro$}, then $(T,\lambda)$ is least resolved if
    and only if all 0-edge are incident to leaves and each inner vertex is
    incident to exactly one 0-edge.}

  \rev{If, in addition, $(T,\lambda)$ is minimally resolved, then all
    0-edge are incident to leaves and each inner vertex is incident to
    exactly one 0-edge.}
\end{lemma}
\begin{proof}
  Let $\Rl_T$, $\Ro_T$, $\Rl_{T'}$, and $\Ro_{T'}$ be the relations
  explained by $T$ and $T'$, respectively.  Since $(T,\lambda)$ explains
  $G(\Rl)/\Ro$ , we have $\Rl_T\,=\,\Rl/\Ro$. Moreover, since $\Ro_T$ is
  discrete, no two distinct leaves of $T$ are in relation $\Ro_T$.

  If $\lambda(e)=0$, then contracting the interior edge $e$ clearly
  preserves the property of $\Ro_{T'}$ being discrete.  Since only interior
  edges are allowed to be contracted, we have $L(T)=L(T')$. Therefore,
  $\Ro_T =\Ro_{T'}$ and the 1-edges along any path from $x\in L(T)=L(T')$
  to $y\in L(T)=L(T')$ remains unchanged, and thus $\Rl_T=\Rl_{T'}$.
  Hence, $(T',\lambda')$ explains $G(\Rl)/\Ro$.

  If $G(\Rl)/\Ro$ is connected, then for every 1-edge $e$ there is a pair
  of leaves $x'$ and $x''$ such that $x'\Rl x''$ and $e$ is the only 1-edge
  along the unique path connecting $x'$ and $x''$. Consequently,
  contracting $e$ would make $x'$ and $x''$ non-adjacent w.r.t.\ the
  resulting relation.  Thus no 1-edge can be contracted in $T$ without
  changing $G(\Rl)/\Ro$. 

  \rev{Now assume that $(T,\lambda)$ is a least resolved tree that explains
    the connected graph $G(\Rl)/\Ro$.  By the latter arguments, all
    interior edges of $(T,\lambda)$ must be 1-edges and thus any 0-edge
    must be incident to leaves.  Assume for contradiction, that there is an
    inner vertex $w$ such that for all adjacent leaves $x',x''$ we have
    $\lambda(wx') = \lambda(wx'') = 1$. Thus, for any such leaves we have
    $x'\not\Rl x''$. In particular, any path from $x'$ to any other leaf
    $y$ (distinct from the leaves adjacent to $w$) contains an interior
    1-edge. Thus, $x'\not \Rl y$ for any such leaf of $T$.  However, this
    immediately implies that $x'$ is an isolated vertex in $G(\Rl)/\Ro$; a
    contradiction to the connectedness of $G(\Rl)/\Ro$.  Furthermore, if
    there is an inner vertex $w$ such that for adjacent leaves $x',x''$ it
    holds that $\lambda(wx') = \lambda(wx'') = 0$, then $x'\Ro x''$; a
    contradiction to $\Ro$ being discrete.  Therefore, each inner vertex is
    incident to exactly one 0-edge.
	
    If $G(\Rl)/\Ro$ is connected and $(T,\lambda)$ explains $G(\Rl)/\Ro$
    such that all 0-edge are incident to leaves, then all interior edges
    are 1-edges. As shown, no interior 1-edge can be contracted in $T$
    without changing the corresponding $\Rl$ relation. Moreover, no
    leaf-edge can be contracted since $L(T) = V(G(\Rl)/\Ro)$. Hence,
    $(T,\lambda)$ is least resolved.

    Finally, since any minimally resolved tree is least resolved, the last
    assertions follows from the latter arguments.    }
\end{proof}

\begin{algorithm}[tbp]
\caption{Compute $(T(Q), \lambda)$}
\label{alg:Q}
\begin{algorithmic}[1]
  \REQUIRE $Q$
  \ENSURE $T(Q)$
  \STATE set $T(Q)\leftarrow Q$ 
  \STATE Retain the labels of all leaves of $Q$ in $T(Q)$ and relabel 
			all inner vertices $u$ of $Q$ as $u'$.
  \STATE Label all edges of the copy of $Q$ by $\lambda(e)=1$.
  \STATE For each inner vertex $u'$ of $Q$ add a vertex $u$ to $T(Q)$ and 
		   insert the edge $uu'$. 
  \STATE Label all edges of the form $e=uu'$ with $\lambda(e)=0$. 
\end{algorithmic}
\end{algorithm}

Let $\mathbb{T}$ be the set of all trees with vertex set $X$ but no
edge-labels and $\mathcal{T}$ denote the set of all edge-labeled 0-1-trees
$(T,\lambda)$ with leaf set $X$ such that each inner vertex $w\in W$ has
degree at least $3$ and there is exactly one adjacent leaf $v$ to $w$ with
$\lambda(wv)=0$ while all other edges $e$ in $T$ have label $\lambda(e)=1$.

\begin{figure}[tbp]
\begin{tabular}{lcr}
\begin{minipage}{0.5\textwidth}
\begin{center}
\includegraphics[width=\textwidth]{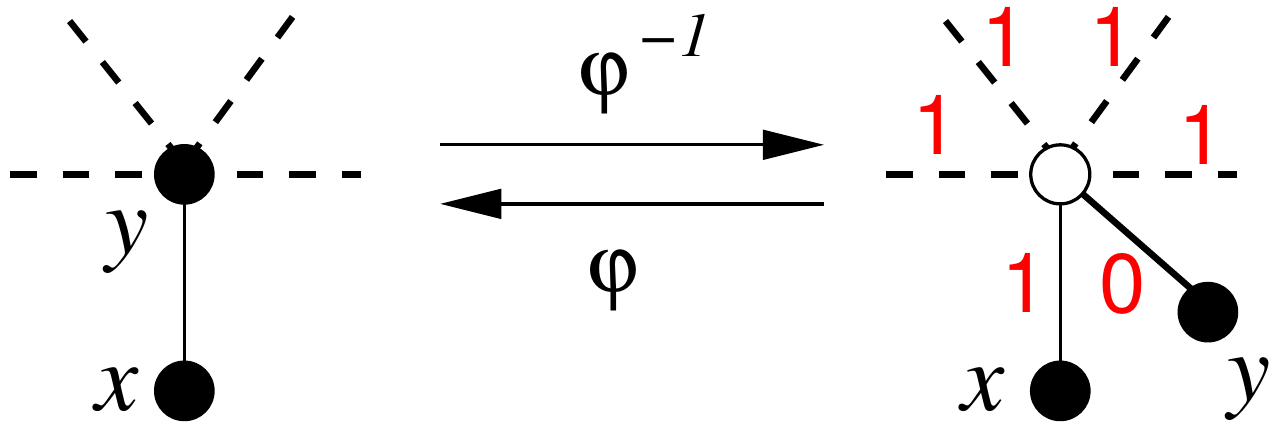}
\end{center} 
\end{minipage} & & 
\begin{minipage}{0.4\textwidth}
  \caption{\rev{Illustration of the bijection $\varphi$. It contracts, at
      each inner vertex (white) of $(T,\lambda)$ the unique $0$-edge and
      transfers $y\in X$ as vertex label at the inner vertex of $Q$.
      Leafs in $(T,\lambda)$ with incident to $1$-edges remain unchanged.
      The inverse map $\varphi^{-1}$ is given by Algorithm~\ref{alg:Q}.}}
  \label{fig:bijection}
\end{minipage}
\end{tabular}
\end{figure}

\rev{
\begin{lemma}
	The map $\varphi : \mathcal{T} \to \mathbb{T}$ with 
        $\varphi: (T,\lambda)\ \mapsto Q$,
	$V(Q)=L(T)$, and  
	$Q\simeq T^*$, where $T^*$ is the underlying unlabeled tree 
        obtained from $(T,\lambda)$ by contracting all edges labeled $0$, 
        is a bijection.
	\label{lem:bijection}
\end{lemma}
\begin{proof}
	We show first that $\varphi$ and $\varphi^{-1}$ are maps.  
	Clearly, $\varphi$ is a map, since the edge-contraction 
	is well-defined and leads to exactly one tree in $\mathbb{T}$. 
	For $\varphi^{-1}$ we construct $(T,\lambda)$ from $Q$
	as in Algorithm \ref{alg:Q}. 
	It is easy to see that  $(T,\lambda) \in  \mathcal{T}$. 
	Now consider $T^*$ obtained from  $(T,\lambda)$ by 
	contracting all edges labeled $0$. By construction, 
	$T^* \simeq Q$ (see Fig. \ref{fig:bijection}). 
	Hence,  $\varphi : \mathcal{T} \to  \mathbb{T}$ is bijective. 
\end{proof}
}

\rev{The bijection is illustrated in Fig.~\ref{fig:bijection}.}

\rev{
\begin{lemma}
  Let $(T,\lambda) \in \mathcal{T}$ and $Q= \varphi((T,\lambda)) \in
  \mathbb{T}$.  The set $\mathcal{T}$ contains \emph{all} least resolved
  trees that explain $Q$.

  Moreover, if $Q$ is considered as a graph \mbox{$G(\Rl)/\Ro$} with vertex
  set $X$, then $(T,\lambda)$ is the unique least resolved tree 
  that
  explains $Q$ and therefore, the unique minimally resolved tree that
  explains $Q$.
  \label{lem:Umin}
\end{lemma}
\begin{proof}
  We start with showing that $(T,\lambda)$ explains $Q$.  Note, since
  $Q\in\mathbb{T}$, the graph $Q$ must be connected.  By construction and
  since $Q= \varphi((T,\lambda))$, $T^* \simeq Q$ where $T^*$ is the tree
  obtained from $T$ after contracting all 0-edges. Let $v,w\in X$ and
  assume that there is exactly a single $1$ along the path from $v$ to $w$
  in $(T,\lambda)$. Hence, after contracting all edges labeled $0$ we see
  that $vw\in E(T^*)$ where $T^* \simeq Q$ and thus $v\Rl w$. Note, no path
  between any two vertices in $(T,\lambda)$ can have only 0-edges (by
  construction).  Thus, assume that there is more than a single 1-edge on
  the path between $v$ and $w$. Hence, after after contracting all edges
  labeled $0$ we see that there is still a path in the tree $T^*$ from $v$
  to $w$ with more than one 1-edge. Since $T^* \simeq Q$, we have $vw\not
  \in E(Q)$ and therefore, $v\not\Rl w$. Thus, $(T,\lambda)$ explains $Q$.

  By construction of the trees in $\mathcal{T}$ all 0-edges are incident to
  a leaf. Thus, by Lemma \ref{lem:contract}, every least resolved tree that
  explains $Q$ is contained in $\mathcal{T}$.

  It remains to show that the least resolved tree $(T,\lambda)\in
  \mathcal{T}$ with $T^*\simeq Q$ that explains $Q$ is minimally resolved.
  Assume there is another least resolved tree $(T',\lambda') \in
  \mathcal{T}$ with leaf set $V(Q)$ that explains $Q$.  By Lemma
  \ref{lem:bijection}, there is a bijection between those $Q$ and elements
  in $\mathcal{T}$ for which $T^*\simeq Q$.  Thus, $T'^*\simeq Q'\not\simeq
  Q \simeq T^*$. However, this implies that $T'\not\simeq T$. However,
  since in this case $(T',\lambda')$ explains $Q'$ and $Q'\not\simeq Q$,
  the pair $(T',\lambda')$ cannot explain $Q$; a contradiction. 
\end{proof}	
}

\rev{As an immediate consequence of these considerations we obtain}
\begin{thm}
  Let $Q$ be a connected component in $G(\Rl)/\Ro$ with vertex set $X$.
  Then the tree $(T,\lambda)$ constructed in Algorithm~\ref{alg:Q} is the unique
  \rev{minimally} resolved tree that explains $Q$. 

  \rev{Moreover, for any pair $(T',\lambda')$ that explains $Q$, the tree
    $T$ is obtained from $T'$ by contracting all \emph{interior} 0-edges
    and putting $\lambda(e) = \lambda'(e)$ for all edges that are not
    contracted. }
  \label{thm:connComp}
\end{thm}
\begin{proof}
  \rev{The first statement follows from Lemma \ref{lem:contract} and
    \ref{lem:Umin}.  To see the second statement, observe that Lemma
    \ref{lem:contract} implies that no interior 1-edge but every 0-edge can
    be contracted.  Hence, after contracting all 0-edges, no edge can be
    contracted and thus, the resulting tree is least resolved. By Lemma
    \ref{lem:contract}, we obtain the result.}   
\end{proof}

\rev{ We emphasize that although the minimally resolved tree that explains
  $G(\Rl)/\Ro$ is unique, this statement is in general not satisfied for
  least resolved trees, see Figure \ref{fig:nonU-LRT}.  }

\begin{figure}[tbp]
\begin{center}
\includegraphics[width=1\textwidth]{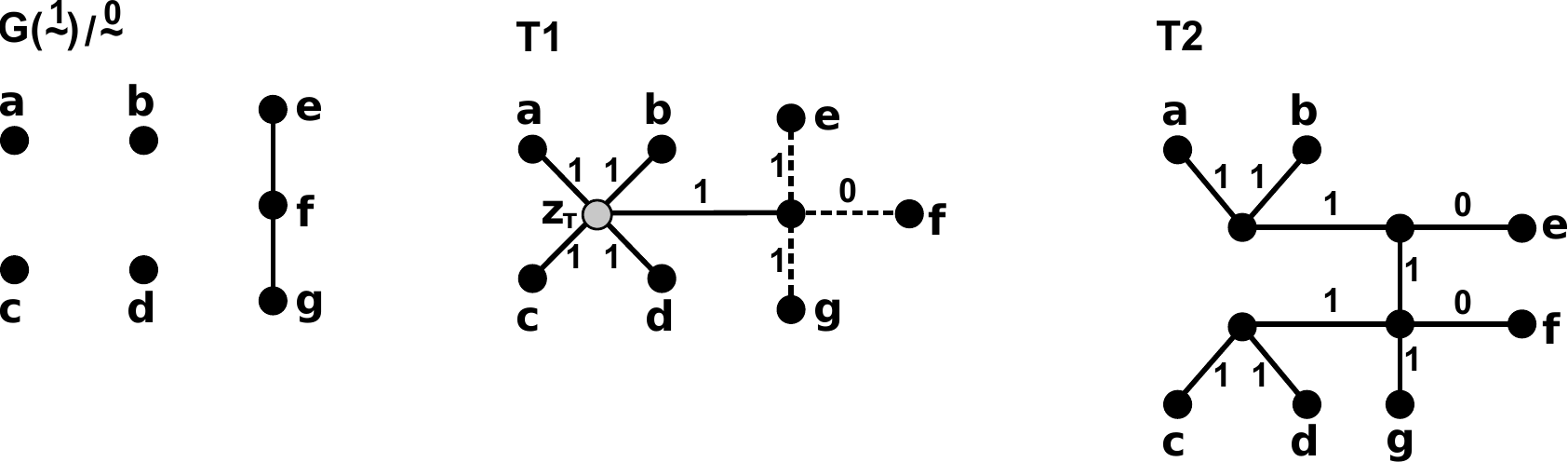}
\end{center}
\caption{\rev{Least resolved trees explaining a given relation are unique whenever
			$G(\Rl)/\Ro$ is connected (cf.\ Lemma \ref{lem:Umin}) and are, therefore, 
			also minimally resolved. 
		 Now,  consider the disconnected graph $G(\Rl)/\Ro$ shown on the
    left. Both pairs $(T_1,\lambda_1)$ (middle) and $(T_2,\lambda_2)$
    (right) are least resolved trees that explain $G(\Rl)/\Ro$.   Thus, uniqueness of least resolved
	trees that explain $G(\Rl)/\Ro$ is not always satisfied.
	 The tree  $(T_1,\lambda_1)$ is obtained with Alg.\ \ref{alg:all}, 
	  by connecting the vertex $z_T$ via 1-edges to the (inner) vertices 
	  of the respective minimally resolved trees $T(Q)$ (the dashed subtree and the single vertex trees
		$a,b,c,d$) that explain the connected components $Q$ of $G(\Rl)/\Ro$.
		In this example, 	 $(T_1,\lambda_1)$ is the unique minimally resolved tree that explains $G(\Rl)/\Ro$, 
		since each tree $T(Q)$ has a unique interior vertex and due to Thm.\ \ref{thm:star-tree}.
	} }
\label{fig:nonU-LRT}
\end{figure}

We are now in the position to demonstrate how to obtain a least resolved
tree that explains $G(\Rl)/\Ro$ also in the case that $G(\Rl)/\Ro$ itself
is not connected. To this end, denote by $Q_1,\dots Q_k$ the connected
components of $G(\Rl)/\Ro$. We can construct a \rev{phylogenetic tree
  $T(G(\Rl)/\Ro)$ with leaf set $X$} for $G(\Rl)/\Ro$ using Alg.\
\ref{alg:all}. It basically amounts to constructing a star $S_k$ with
inner vertex $z$, where its leaves are identified with the trees
$T(Q_i)$.

\begin{algorithm}[tbp]
\caption{Compute $(T(G(\Rl)/\Ro)), \lambda)$}
\label{alg:all}
\begin{algorithmic}[1]
  \REQUIRE disconnected $G(\Rl)/\Ro)$
  \ENSURE $T(G(\Rl)/\Ro))$	
  \STATE $T(G(\Rl)/\Ro)) \gets (\{\rev{z_T}\},\emptyset)$
  \FOR{For each connected component $Q_i$}
     \STATE	construct $(T(Q_i), \lambda_i)$ with Alg.\ \ref{alg:Q}
		and add to $T(G(\Rl)/\Ro))$.
     \IF{$T(Q_i)$ is the single vertex graph $(\{x\},\emptyset)$}
	\STATE add edge $\rev{z_T}x$
     \ELSIF{$T(Q_i)$ is the edge $v_iw_i$}
        \STATE remove the edge
                $v_iw_i$ from $T(Q_i)$, insert a vertex $x_i$ in $T(Q_i)$
                and the edges $x_iv_i$, $x_iw_i$. 
        \STATE  set either $\lambda_{i}(x_iv_i)=1$ and $\lambda(x_iw_i)=0$
		or $\lambda_{i}(x_iv_i)=0$ and $\lambda(x_iw_i)=1$. 
                \label{step:edge}
	\STATE  add edge $\rev{z_T}x_i$ to $T(G(\Rl)/\Ro))$.
     \ELSE \STATE \label{item:z}  	
                add edge $\rev{z_T}q'_i$ to  $T(G(\Rl)/\Ro))$
		for an arbitrary inner vertex $q'_i$ of $T(Q_i)$. 
     \ENDIF		 
   \ENDFOR
   \STATE       Set $\lambda(\rev{z_T}v)=1$ for all edges $\rev{z_T}v$ and 
		$\lambda(e)=\lambda_i(e)$ for all edges $e\in T(Q_i)$.	
\end{algorithmic}
\end{algorithm}

\rev{
\begin{lemma}
  Let $G(\Rl)/\Ro$ have connected components $Q_1,\dots Q_k$.  Let $T'$ be
  a tree that explains $G(\Rl)/\Ro$ and $T'_i$ be the subtree of $T'$ with
  leaf set $V(Q_i)$ that is minimal w.r.t.\ inclusion, $1\le i\le k$.
  Then, $V(T'_i)\cap V(T'_j) =\emptyset$, $i\neq j$ and, in particular, any
  two vertices in $T'_i$ and $T'_j$, respectively, have distance at least
  two in $T'$.
\label{lem:subtrees}
\end{lemma}
\begin{proof}
  We start to show that two distinct subtrees $T'_i$, and $T'_j$ do not
  have a common vertex in $T'$.  If one of $Q_i$ or $Q_j$ is a single
  vertex graph, then $T'_i$ or $T'_j$ consists of a single leaf only, and
  the statement holds trivially.

  Hence, assume that both $Q_i$ and $Q_j$ have at least three vertices.
  Lemma \ref{lem:contract} implies that each inner vertex of the minimally
  resolved trees $T(Q_i)$ and $T(Q_j)$ is incident to exactly one 0-edge as
  long as $Q_i$ is not an edge.  Since $T(Q_i)$ can be obtained from $T'_i$
  by the procedure above, for each inner vertex $v$ in $T'_i$ there is a
  leaf $x$ in $T'_i$ such that the unique path from $v$ to $x$ contains
  only 0-edges.  The same arguments apply, if $Q_i$ is an edge $xy$. In
  this case, the tree $T'_i$ must have $x$ and $y$ as leaves, which implies
  that $T'_i$ has at least one inner vertex $v$ and that there is exactly
  one 1-edge along the path from $x$ to $y$. Thus, for each inner vertex in
  $T'_i$ there is a path to either $x$ or $y$ that contains only
  0-edges.

  Let $v$ and $w$ be arbitrary inner vertices of $T'_i$ and $T'_j$,
  respectively, and let $x$ and $y$ be leaves that are connected to $v$ and
  $w$, resp., by a path that contains only 0-edges. If $v=w$, then $x\Ro
  y$, contradicting the property of $\Ro$ being discrete.  Thus, $T'_i$ and
  $T'_j$ cannot have a common vertex in $T'$.  Moreover, there is no edge
  $vw$ in $T'$, since otherwise either $x\Ro y$ (if $\lambda(vw)=0$) or
  $x\Rl y$ (if $\lambda(vw)=1$).  Hence, any two distinct vertices in
  $T'_i$ and $T'_j$ have distance at least two in $T'$. 
\end{proof}
}

\rev{We note that Algorithm~\ref{alg:all} produces a tree with a single
  vertex of degree $2$, namely $z_T$, whenever $G(\Rl)/\Ro$ consists of
  exactly two components. Although this strictly speaking violates the
  definition of phylogenetic trees, we tolerate this anomaly for the
  remainder of this section.}
  
\begin{thm}
  Let $Q_1,\dots Q_k$ be the connected components in $G(\Rl)/\Ro$.  Up to
  the choice of the vertices $q'_i$ in Line \ref{item:z} of Alg.\
  \ref{alg:all}, the tree $T^* = T(G(\Rl)/\Ro))$ is a \rev{minimally
    resolved tree that explains $G(\Rl)/\Ro$. It is unique up to the choice
    of the $\rev{z_{T}}q'_i$ in Line \ref{item:z}.}
  \label{thm:star-tree}
\end{thm}
\begin{proof}
  \rev{Since every tree $T(Q_i)$ explains a connected component in
    $G(\Rl)/\Ro$, from the construction of $T^*$ it is easily seen that
    $T^*$ explains $G(\Rl)/\Ro$.  Now we need to prove that $T^*$ is a
    minimally resolved tree that explains $G(\Rl)/\Ro$.

    To this end, consider an arbitrary tree $T'$ that explains
    $G(\Rl)/\Ro$.  Since $T'$ explains $G(\Rl)/\Ro$, it must explain each
    of the connected components $Q_1,\dots Q_k$. Thus, each of the subtrees
    $T'_i$ of $T'$ with leaf set $V(Q_i)$ that are minimal w.r.t.\
    inclusion must explain the connected component $Q_i$, $1\le i \le k$.
    Note the $T'_i$ may have vertices of degree 2.

    We show first that $T(Q_i)$ is obtained from $T'_i$ by contracting all
    interior 0-edges and all 0-edges of degree 2.  If there are no vertices
    of degree 2, we can immediately apply Thm.\ \ref{thm:connComp}.
  	
    If there is a vertex $v$ of degree 2, then $v$ cannot be incident to
    two 1-edges, as otherwise the relation explained by $T'_i$ would not be
    connected, contradicting the assumption that $T'_i$ explains the
    connected component $Q_i$.  Thus, if there is a vertex $v$ of degree 2
    it must be incident to a 0-edge $vw$.  Contracting $vw$ preserves the
    property of $\Ro$ being discrete.  If $w$ is a leaf, we can contract
    the edge $vw$ to a new leaf vertex $w$; if $vw$ is an interior edge we
    simply contract it to some new inner vertex. In both cases, we can
    argue analogously as in the proof of Lemma \ref{lem:contract} that the
    tree obtained from $T'_i$ after contracting $vw$, still explains $Q_i$.
    This procedure can be repeated until no degree-two vertices are in the
    contracted $T'_i$.

    In particular, the resulting tree is a phylogenetic tree that explains
    $Q_i$. Now we continue to contract all remaining interior 0-edges.
    Thm.\ \ref{thm:connComp} implies that in this manner we eventually
    obtain $T(Q_i)$.
  
    By Lemma \ref{lem:subtrees}, two distinct tree $T(Q_i)$ and $T(Q_j)$ do
    not have a common vertex, and moreover, any two vertices in $T(Q_i)$
    and $T(Q_j)$, respectively, have distance at least two in $T'$.

    This implies that the construction as in Alg.\ \ref{alg:all} yields a
    least resolved tree.  In more detail, since the subtrees explaining
    $Q_i$ in any tree that explains $G(\Rl)/\Ro$ must be vertex disjoint,
    the minimally resolved trees $T(Q_1),\dots,T(Q_k)$ must be subtrees of
    any minimally resolved tree that explain $G(\Rl)/\Ro$, as long as all
    $Q_i$ are single vertex graphs or have at least one inner vertex.

    If $Q_i$ is a single edge $v_iw_i$ and thus $T(Q_i) = v_iw_i$ where
    $\lambda(v_iw_i)=1$, we modify $T(Q_i)$ in Line \ref{step:edge} to
    obtain a tree isomorphic to $S_2$ with inner vertex $x_i$. This
    modification is necessary, since otherwise (at least one of) $v_i$ or
    $w_i$ would be an inner vertex in $T^*$, and we would loose the
    information about the leaves $v_i,w_i$. In particular, we need to add
    this vertex $x_i$ because we cannot attach the leaves $v_i$ (resp.\
    $w_i$) by an edge $x_jv_i$ (resp.\ $x'_jw_i$) to some subtree subtree
    $T(Q_j)$.  To see this, note that at least one of the edges $x_jv_i$
    and $x'_jw_i$ must be a 0-edge. However, $x_j$ and $x'_j$ are already
    incident to a 0-edge $x_jv'_i$ or $x'_jw'_i$ (cf.\ Lemma
    \ref{lem:contract}), which implies that $\Ro$ would not be discrete; a
    contradiction.  By construction, we still have $v_i\Rl w_i$ in Line
    \ref{step:edge}.
	
    Finally, any two distinct vertices in $T(Q_i)$ and $T(Q_j)$ have
    distance at least two in $T^*$, as shown above. Hence, any path
    connecting two subtrees $T(Q_i)$ in $T^*$ contains and least two edges
    and hence at least one vertex that is not contained in any of the
    $T(Q_i)$. Therefore, any tree explaining $Q$ has at least $1+\sum_i
    |V(T(Q_i))|$ vertices.

    We now show that adding a single vertex \rev{$z_T$}, which we may consider as
    a trivial tree $(\{\rev{z_T}\},\emptyset)$, is sufficient. Indeed, we may
    connect the different trees to $z_T$ by insertion of an edge $\rev{z_T}q'_i$,
    where $q'_i$ is an arbitrary inner vertex of $T(Q_i)$ and label these
    edges $\lambda(\rev{z_T}q'_i)=1$.  Thus, no two leaves $u$ and $w$ of distinct
    trees are either in relation $\Ro$ or $\Rl$, as required. The resulting
    trees have the minimal possible number of vertices, i.e., they are
    minimally resolved.   }
\end{proof}


\subsection*{Binary trees}

Instead of asking for least resolved trees that explain $G(\Rl)/\Ro$, we may
also consider the other extreme and ask which binary, i.e., fully resolved
tree can explain $G(\Rl)/\Ro$. Recall that an $X$-tree is called binary or
fully resolved if the root has degree $2$ while all other inner vertices
have degree $3$.  From the construction of the least resolved trees we
immediately obtain the following:
\begin{cor}
  A least resolved tree $T(Q)$ for a connected component $Q$ of
  $G(\Rl)/\Ro$ is binary if and only if $Q$ is a path.
\end{cor}

If a least resolved tree $T(Q)$ of $G(\Rl)/\Ro$ is a star, we have:
\begin{lemma}\label{lem:star}
  If a least resolved tree $T(Q)$ explaining $G(\Rl)/\Ro$ is a star
  with $n$ leaves, then either
  \begin{itemize}
    \item[(a)] all edges in $T(Q)$ are 1-edges and $Q$ has no edge, or 
    \item[(b)] there is exactly one 0-edge in $T(Q)$ and $Q$ is a star
        with $n-1$ leaves.
    \end{itemize}
  \end{lemma}
\begin{proof}
  For implication in case (a) and (b) we can re-use exactly the same
  arguments as in the proofs of Theorem \ref{thm:connComp} and
  \ref{thm:star-tree}.

  Now suppose there are at least two (incident) 0-edges in $T(Q)$, whose
  endpoints are the vertices $u$ and $v$. Then $u\Ro v$, which is
  impossible in $G(\Rl)/\Ro$.   
\end{proof}

\rev{
\begin{lemma}\label{lem:local-star}
  Let $(T,\lambda)$ be a least resolved tree that explains $G(\Rl)/\Ro$.
  Consider an arbitrary subgraph $S_k$ that is induced by an inner vertex
  $v_0$ and \emph{all} of its $k$ neighbors.  Then, $S_k$ with its
  particular labeling $\lambda_{|E(S_k)}$ is always of type (a) or (b) as
  in Lemma \ref{lem:star}.
\end{lemma}
\begin{proof}
 This is an immediate consequence of Lemma \ref{lem:contract}
 and the fact that $\Ro$ is discrete.  
\end{proof}
}

To construct the binary tree \rev{explaining} the star $Q=S_n$, we consider
the set of all binary trees with $n$ leaves and 0/1-edge labels. If $S_n$
is of type (a) in Lemma~\ref{lem:star}, then all terminal edges are labeled
$1$ and all interior edges are arbitrarily labeled $0$ or $1$.
Figure~\ref{fig:non_lrt} shows an example for $S_6$.  If $S_n$ is of type
(b), we label the terminal edges in the same way as in $T(Q)$ and all
interior edge are labeled 0. In this case, for each binary tree there is
exactly one labeling.

\begin{figure}[htbp]
\begin{center}
\includegraphics[width=0.6\textwidth]{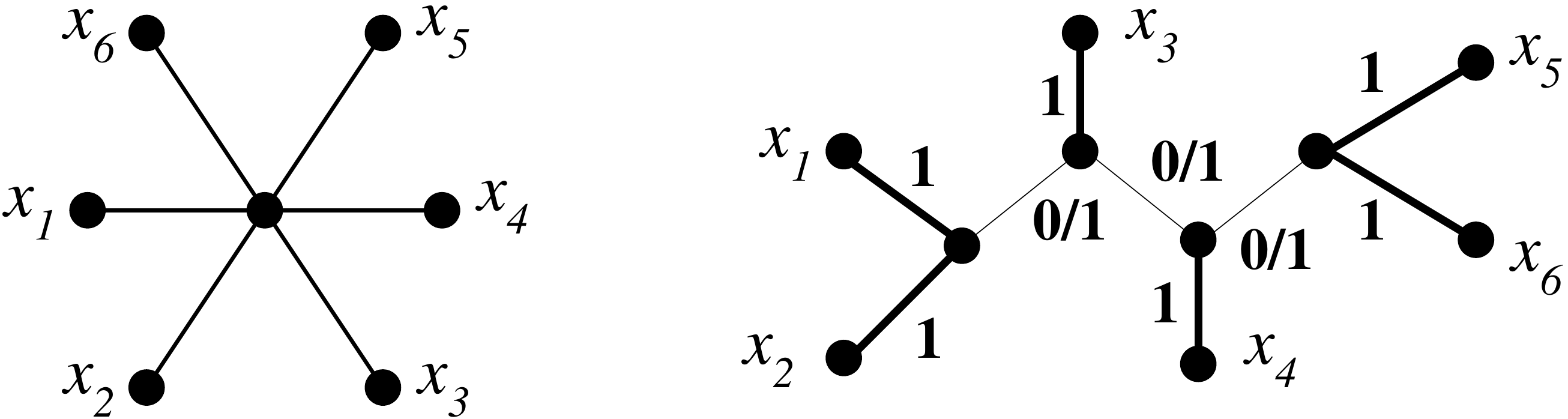}
\end{center}
\caption{For fixed underlying tree $T(Q)$, in this case a star $S_6$ with
  all $1$, there are in general multiple labelings $\lambda$ that
  \rev{explains} the same relation $Q$, here the empty relation.}
\label{fig:non_lrt}
\end{figure}

In order to obtain the complete set of binary trees that explain $G$ we can
proceed as follows. If $G$ is connected, there is a single minimally resolved
tree $T(G)$ explaining $G$. If $G$ is not connected then there are multiple
minimally resolved trees $T$. 

\rev{Let $\mathcal{T}_{\mathrm{lrt}}$ be the set of all least resolved
  trees that explain $G$. For every such least resolved tree $T\in
  \mathcal{T}_{\mathrm{lrt}}$} we iterate over all vertices $v_0$ of $T$
with degree $k>3$ and perform the following manipulations:
\begin{itemize}
\item[{1.}] Given a vertex $v_0$ of $T$ with degree $k>3$, denote the set of 
	its  neighbors $v_1, v_2,\dots, v_k$  by $N(v_0)$. 
	Delete vertex $v_0$ and its attached
  edges from $T$, and rename the neighbors $v_i$ to $v_i'$ for all $1\leq
  i\leq k$. Denote the resulting forest by $F(v_0)$.
\item[{2.}] Generate all binary trees with leaves $v_1',\dots, v_k'$. 
\item[{3.}] Each of these binary trees is inserted into a copy of the
  forest $F(v_0)$ by \rev{identifying} $v_i$ and $v_i'$ for all $1\leq
  i\leq k$.
\item[{4.}] \rev{ For each of the inserted binary trees $T'$ that results
    from a ``local'' star $S_k$ in step 3. we must place an edge label. \\
    Put $\lambda(xv_i)= \lambda(v_0v_i)$ for all edges $xv_i$ in
    $T'$ and mark  $xv_i$ as \emph{LABELED}.
 	If $S_k$ is of type (a) (cf.\, Lemma~\ref{lem:star}), then
    choose an arbitrary 0/1-label for the interior edges of $T'$.  If $S_k$
    is of type (b) we need to consider the two exclusive cases for the
    vertex $v_j$ for which $\lambda(v_jx)=0$:
    \begin{itemize}
    \item[(i)] For all  $y\in V(G)$ for which $v_j \Rl y$, label all \emph{interior}
      edges on the unique path $\mathbb{P}(v_j,y)$ that are also contained in $T'$  and are not marked 
		as  \emph{LABELED} with 0 and mark them as   \emph{LABELED}
		and
      choose an arbitrary 0/1-label for all other un\emph{LABELED} interior edges of $T'$.
    \item[(ii)] Otherwise, choose an arbitrary 0/1-label for the interior
      edges of $T'$.
    \end{itemize}
}
\end{itemize}
\rev{It is well known that each binary tree has $k-3$ interior edges
  \cite{sem-ste-03a}.  Hence, for a binary tree there are $2^{k-3}$
  possibilities to place a 0/1 label on its interior edges.  Let $t(k)$
  denote the number of binary trees with $k$ leaves and $V_a,V_b$ be a
  partition of the inner vertices into those where the neighborhood
  corresponds to a star of type (a) and (b), respectively. Note, if $T$ is
  minimally resolved, then $|V_a|\le1$.  For a given least resolved tree
  $T\in \mathcal{T}_{\mathrm{lrt}}$, the latter procedure yields the set of
  all $ (\prod_{v\in V_a} t(\mathrm{deg}(v_0))2^{\mathrm{deg}(v_0)-3} )
  \cdot (\prod_{v\in V_b} t(\mathrm{deg}(v_0) )) $ pairwise distinct binary
  trees that one can obtain from $T$. The union of these tree sets and its
  particular labeling over all $T\in \mathcal{T}_{\mathrm{lrt}}$ is then
  the set of all binary trees explaining $G$.  } To establish the
correctness of this procedure, we prove \newline
\begin{lemma}\label{lem:binary}  
  The procedure outlined above generates all binary trees $(T,\lambda)$
  \rev{explaining} $G$.
\end{lemma}

\rev{
\begin{proof}
  We first note that there may not be a binary tree explaining $G$. This is
  case whenever $T(G)$ has a vertex of degree $2$, which is present in
  particular if $G$ is forest with two connected components.

  Now consider an arbitrary binary tree $(T_B(G), \lambda)$ that is not
  least resolved for $G$. Then a least resolved tree $T(G)$ explaining $G$
  can be obtained from $T_B(G)$ by contracting edges and retaining the the
  labeling of all non-contracted edges. In the following we will show that
  the construction above can be be used to recover $T_B(G) = (V,E)$ from
  $T(G)$. To this end, first observe that only interior edges can be
  contracted in $T_B(G)$ to obtain $T(G)$.  Let $E' = \{e_1,\dots,e_h\}$ be
  a maximal (w.r.t.\ inclusion) subset of contracted edges of $T_B(G)$ such
  that the subgraph $(V' = \cup_{i=1}^h e_i, E')$ is connected, and thus
  forms a subtree of $T_B(G)$.  Furthermore, let $F = \{f_1,\dots,f_k\}
  \subseteq E\setminus E'$ be a maximal subset of edges of $T_B(G)$ such
  that for all $f_i\in F$ there is an edge $e_j\in E'$ such that $f_i\cap
  e_j\neq \emptyset$.  Moreover, set $W = \cup_{i=1}^k f_i$.  Thus, the
  contracted subtree $(V',E')$ locally corresponds to the vertex $v_0$ of
  degree $k>3$ and thus, to a local star $S_k$.  Now, replacing $S_k$ by
  the tree $(V'\cup W, E'\cup F)$ (as in Step 3) yields the subtree from
  which we have contracted all interior edges that are contained in $E'$.
  Since the latter procedure can be repeated for all such maximal sets
  $E'$, we can recover $T_B(G)$.

  It remains to show that one can also recover the labeling $\lambda$ of
  $T_B(G)$. Since $T(G)$ is a least resolved tree obtained from $T_B(G)$
  that explains $G$, we have, by definition, $u \Rl w$ in $T(G)$ if and
  only if $u \Rl w$ in $T_B(G)$.  By Lemma \ref{lem:local-star}, every
  local star $S_k$ in $T(G)$ is either of type (a) or (b).  Assume it is of
  type (a), i.e., $\lambda(v_0v_i)=1$ for all $1\leq i\leq k$.  Let $u$ and
  $w$ be leaves of $G$ for which the unique path connecting them contains
  the edge $v_0v_i$ and $v_0v_j$.  Thus, there are at least two edges
  labeled $1$ along the path; hence $u\not\Rl w$. The edges $xv_i$ and
  $x'v_j$ in $T'$ are both labeled by 1; therefore $u$ and $w$ are not in
  relation $\Rl$ after replacing $S_k$ by $T'$.  Therefore all possible
  labelings can be used in Step 4 except for the edges that are not
  contained in $T'$ and $xv_i$ and $x'v_j$ which are marked as
  \emph{LABELED}). Therefore, we also obtain the given labeling of the
  subtree $T = (V'\cup W, E'\cup F)$ as a result.

  If the local star $S_k$ is of type (b), then there is exactly one edge
  $v_0v_j$ with $\lambda(v_0v_j) = 0$. By Lemma \ref{lem:contract} and
  because $(T(G),\lambda)$ is least resolved, the vertex $v_j$ must be a
  leaf of $G$.  For two leaves $u,w$ of $T(G)$ there are two cases: either
  the unique path $\mathbb{P}(u,w)$ in $T(G)$ contains $v_0$ or not.
	
  If $\mathbb{P}(u,w)$ does not contain $v_0$, then this path and its
  labeling remains unchanged after replacing $S_k$ by $T'$.  Hence, the
  relations $\Rl$ and $\not \Rl$ are preserved for all such vertices $u,w$.

  First, assume that $\mathbb{P}(u,w)$ contains $v_0$ and thus two edges
  incident to vertices in $N(v_0)$.  If the path $\mathbb{P}(u,w)$ contains
  two edges $v_0v_i$ and $v_0v_l$ with $i,l\neq j$, then $\lambda(v_0v_i) =
  \lambda(v_0v_l)=1$.  Thus, $u\not\Rl w$ in $T(G)$.  The extended path
  (after replacing $S_k$ by $T'$) still contains the two 1-edges $xv_i$ and
  $x'v_l$, independently from the labeling of all other edges in $T'$ that
  have remained un\emph{LABELED} edges up to this point. Thus, $u\not\Rl w$
  is preserved after replacing $S_k$ by $T'$.

  Next, assume that $\mathbb{P}(u,w)$ contains $v_0$ and the 0-edge
  $v_0v_j$ in $T(G)$. In the latter case, $u=v_j$ of $w=v_j$.  Note, there
  must be another edge $v_0v_l$ in $\mathbb{P}(u,w)$ with $v_l \in N(v_0)$,
  $l\neq j$ and therefore, with $\lambda(v_0v_l)=1$.  There are two cases,
  either $u\Rl w$ or $u\not\Rl w$ in $T(G)$.
	
  If $u\Rl w$ then there is exactly one 1-edge (the edge $v_0v_l$)
  contained in $\mathbb{P}(u,w)$ in $T(G)$. By construction, all interior
  edges on the path $\mathbb{P}(u,w)$ that are contained in $T'$ are
  labeled with $0$ and all other edge-labelings remain unchanged in $T(G)$
  after replacing $S_k$ by $T'$. Thus, $u\Rl w$ in $T(G)$ after replacing
  $S_k$ by $T'$.  Analogously, if $u\not\Rl w$, then there are at least two
  edges 1-edges $\mathbb{P}(u,w)$ in $T(G)$.  Since $\lambda(v_0v_j)=0$ and
  $\lambda(v_0v_l)=1$, the 1-edge different from $v_0v_l$ is not contained
  in $T'$ and its label 1 remains unchanged.  Moreover, the edge $xv_l$ in
  $T'$ gets also the label 1 in Step 4.  Thus, $\mathbb{P}(u,w)$ still
  contains at least two 1-edges in $T(G)$ after replacing $S_k$ by $T'$
  independently of the labeling chosen for the other un\emph{LABELED}
  interior edges of $T'$.  Thus, $u\not\Rl w$ in $T(G)$ after replacing
  $S_k$ by $T'$.
	
  We allow all possible labelings and fix parts where necessary. In
  particular, we obtain the labeling of the subtree $(V'\cup W, E'\cup F)$
  that coincides with the labeling of $T_B(G)$. Thus, we can repeat this
  procedure for all stars $S_k$ in $T(G)$ and their initial labelings.
  Therefore, we can recover both $T_B(G)$ and its edge-labeling $\lambda$.
  Clearly, every binary tree $T_B(G)$ that explains $G$ is either already
  least resolved or there is a least tree $T(G)\in
  \mathcal{T}_{\mathrm{lrt}}$ from which $T_B(G)$ can be recovered by the
  construction as outlined above.  
\end{proof}
}

\rev{As a consequence of} the proof of Lemma~\ref{lem:binary} we
immediately obtain the following Corollary that characterizes the condition
that $Q$ cannot be explained by a binary tree.
\begin{cor}
  $G(\Rl)/\Ro$ cannot be explained by a binary tree if and only if
  $G(\Rl)/\Ro$ \rev{is a forest with} exactly two connected components.
\end{cor}
The fact that exactly two connected components appear as a special case is
the consequence of a conceptually too strict definition of ``binary
tree''. If we allow a single ``root vertex'' of degree $2$ in this special
case, we no longer have to exclude two-component graphs.

\section{The antisymmetric single-1 relation}
\label{sect:1dir}

The antisymmetric version $x\Rld y$ of the 1-relation shares many basic
properties with its symmetric cousin. We therefore will not show all formal
developments in full detail. Instead, we will where possible appeal to the
parallels between $x\Rld y$ and $x\Rl y$. For convenience we recall the
definition: $x \Rld y$ if and only if all edges along $\mathbb{P}(u,x)$ are
labeled $0$ and exactly one edge along $\mathbb{P}(u,y)$ is labeled $1$,
where $u=\lca{x,y}$. As an immediate consequence we may associate with
$\Rld$ a symmetrized 1-relation $x\Rl y$ whenever $x\Rld y$ or $y\Rld
x$. Thus we can infer (part of) the underlying unrooted tree topology by
considering the symmetrized version $\Rl$. On the other hand, $\Rld$ cannot
convey more information on the unrooted tree from which $\Rld$ and its
symmetrization $\Rl$ are derived. It remains, however, to infer the
position of the root from directional information. Instead of the
quadruples used for the unrooted trees in the previous section, structural
constraints on rooted trees are naturally expressed in terms of triples.

In the previous section we have considered $\Rl$ in relation to unrooted
trees only. Before we start to explore $\Rld$ we first ask whether $\Rl$
contains any information about the position of the root and if it already
places any constraints on $\Rld$ beyond those derived for $\Rl$ in the
previous section. In general the answer to this question will be negative,
as suggested by the example of the tree $T_5^*$ in Figure
\ref{fig:root}. Any of its inner vertex can be chosen as the root, and
each choice of a root vertex yields a different relation $\Rld$.

Nevertheless, at least partial information on $\Rld$ can be inferred
uniquely from $\Rl$ and $\Ro$. Since all connected components in
$G(\Rl)/\Ro$ are trees, we observe that the underlying graphs
$\underline{G(\Rld)/\Ro}$ of all connected components in $G(\Rld)/\Ro$ must
be trees as well.  Moreover, since $\Ro$ is discrete in $G(\Rl)/\Ro$, it is
also discrete in $G(\Rld)/\Ro$.

Let $Q$ be a connected component in ${G(\Rld)/\Ro}$.  \rev{If $Q$ is an
  isolated vertex or a single edge, there is only a single phylogenetic
  rooted tree (a single vertex and a tree with two leaves and one inner
  root vertex, resp.)  that explains $Q$ and the position of its root is
  uniquely determined.}

Thus we assume that $Q$ contains at least three
vertices from here on. By construction, any three vertices $x,y,z$ in a
connected component $Q$ in $G(\Rld)/\Ro$ either induce a disconnected
graph, or a tree on three vertices.  Let $x,y,z\in V(Q)$ induce such a
tree. Then there are three possibilities (up to relabeling of the vertices)
for the induced subgraph contained in $G(\Rld)/\Ro = (V,E)$:
\begin{itemize}
\item[(i)] $xy, yz \in E$ implying that $x\Rld y \Rld z$, 
\item[(ii)] $yx, yz \in E$ implying that $y\Rld x$ and $y \Rld z$,
\item[(iii)] $xy, zy \in E$ implying that $x\Rld y$ and $z \Rld y$. 
\end{itemize}
Below, we will show that Cases (i) and (ii) both imply a unique tree on the
three leaves $x,y,z$ together with a unique 0/1-edge labeling for the
unique resolved tree $T(Q)$ that displays $Q$, see Fig.\
\ref{fig:2cases}. Moreover, we shall see that Case (iii) cannot occur.

\begin{figure}[t]
\begin{center}
\includegraphics[width=\textwidth]{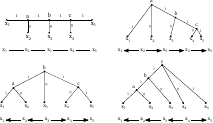}
\end{center}
\caption{Placing the root. The tree $T$ in the upper left is the unique
  minimally resolved tree that explains $G(\Rl)/\Ro$ (shown below $T$).
  Each of the tree inner vertices $a$, $b$, or $c$ of $T$ can be chosen as
  the root, giving rise to three distinct relations $\Rld$.  For the
  ``siblings'' in the unrooted tree $x_1,x_2$ as well as $x_4,x_5$ it holds
  that $x_2\Rld x_1$ and $x_4\Rld x_5$ for all three distinct relations.
  Thus, there are uniquely determined parts of $\Rld$ conveyed by the
  information of $\Ro$ and $\Rl$ only.}
\label{fig:root}
\end{figure} 
\begin{figure}[t]
\begin{center}
\includegraphics[width=0.6\textwidth]{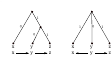}
\end{center}
\caption{There are only two possibilities for induced connected subgraphs
  $H$ in $G(\Rld)/\Ro$ on three vertices, cf.\ Lemma \ref{lem:case-i},
  \ref{lem:case-ii}, and \ref{lem:case-iii}.  Each of the distinct induced
  subgraphs imply a unique least resolved tree with a unique labeling.}
\label{fig:2cases}
\end{figure} 

\begin{lemma}
  In Case (i), the unique triple $\rt{yz|x}$ must be displayed by
  any tree $T(Q)$ that explains $Q$.  Moreover, the paths $\P(u,v)$ and
  $\P(v,z)$ in $T(Q)$ contain both exactly one 1-edge, while the other
  paths $\P(u,x)$ and $\P(v,y)$ contain only 0-edges, where
  $u=\lca{xy}=\lca{xz} \neq \lca{yz}=v$.
  \label{lem:case-i}
\end{lemma}
\begin{proof}
  Let $x,y,z\in V(Q)$ such that $xy, yz \in E$ and thus, $x\Rld y \Rld z$.
  Notice first that there must be two distinct last common ancestors for
  pairs of the three vertices $x,y,z$; otherwise, if
  $u=\lca{xy}=\lca{xz} = \lca{yz}$, then the path $\P(uy)$ contains a
  1-edge (since $x\Rld y$) and hence $y\Rld z$ is impossible.  We
  continue to show that $u=\lca{xy}=\lca{xz}$.  Assume that $u=\lca{xy}\neq
  \lca{xz}$.  Hence, either the triple $\rt{xz|y}$ or $\rt{xy|z}$ is
  displayed by $T(Q)$. In either case the path $\P(u,y)$ contains a 1-edge,
  since $x\Rld y$. This, however, implies $y\not\Rld z$, a contradiction.
  Thus, $u=\lca{xy} = \lca{xz}$. Since there are two distinct last common
  ancestors, we have $u\neq v=\lca{yz}$.  Therefore, the triple $\rt{yz|x}$
  must be displayed by $T(Q)$. From $y\Rld z$ we know that  $\P(v,y)$ only
  contains 0-edges and $\P(v,z)$ contains exactly one 1-edge;  $x\Rld y$
  implies that $\P(x,u)$ contains only 0-edges.  Moreover, since $\P(x,y) =
  \P(x,u)\cup \P(u,v)\cup \P(v,y)$ and $x\Rld y$, the path $\P(u,v)$ must
  contain exactly one 1-edge.    
\end{proof}

\begin{lemma}
  In Case (ii), there is a unique tree on the three vertices $x,y,z$
  with single root $\rho$ displayed by any least resolved tree $T(Q)$ that
  explains $Q$.  Moreover, the path $\P(\rho,y)$ contains only 0-edges,
  while the other paths $\P(\rho,x)$ and $\P(\rho,z)$ must both contain
  exactly one 1-edge.
  \label{lem:case-ii}
\end{lemma}
\begin{proof}
  Assume for contradiction that there is a least resolved tree $T(Q)$
  that displays $\rt{xy|z}$, $\rt{yz|x}$, or $\rt{xz|y}$. 

  The choice of $\rt{xy|z}$ implies $u=\lca{xy}\neq \lca{xz}=\lca{yz}=v$.
  Since $y\Rld x$ and $y\Rld z$, $\P(v,y) \subsetneq \P(u,y)$ contain only
  0-edges, while $\P(u,x)$ and $\P(v,z)$ each must contain exactly one
  1-edge, respectively.  This leads to a tree $T'$ that yields the correct
  $\Rld$-relation.  However, this tree is not least resolved.  By
  contracting the path $\P(u,v)$ to a single vertex $\rho$ and maintaining
  the labels on $\P(\rho,x)$, $\P(\rho,y)$, and $\P(\rho,z)$ we obtain the
  desired labeled least resolved tree with single root.
	
  For the triple $\rt{yz|x}$ the existence of the unique, but not least
  resolved tree can be shown by the same argument with exchanged roles of
  $x$ and $y$.
	
  For the triple $\rt{xz|y}$ we $u = \lca{xy} = \lca{yz} \neq \lca{xz} =
  v$. From $y\Rld x$ and $y\Rld z$ we see that both paths $\P(u,x) =
  \P(u,v)\cup \P(v,x)$ and $\P(u,z) = \P(u,v)\cup \P(v,z)$ contain exactly
  one 1-edge, while all edges in $\P(u,y)$ are labeled $0$.  There are two
  cases: (1) The path $\P(u,v)$ contains this 1-edge, which implies that
  both paths $\P(v,x)$ and $\P(v,z)$ contain only 0-edges. But then $x\Ro
  z$, a contradiction to $\Ro$ being discrete.  (2) The path $\P(u,v)$
  contains only 0-edges, which implies that each of the paths $\P(v,x)$ and
  $\P(v,z)$ contain exactly one 1-edge. Again, this leads to a tree that
  yields the correct $\Rld$-relation, but it is not least resolved.  By
  contracting the path $\P(u,v)$ to a single vertex $\rho$ and maintaining
  the labels on $\P(\rho,x)$, $\P(\rho,y)$, and $\P(\rho,z)$ we obtain the
  desired labeled least resolved tree with single root.   
\end{proof}

\begin{lemma}
  Case (iii) cannot occur.
  \label{lem:case-iii}
\end{lemma}
\begin{proof}
  Let $x,y,z\in V(Q)$ such that $xy, zy \in E$ and thus, $x\Rld y$ and $z
  \Rld y$.  Hence, in the rooted tree that \rev{explains} this relationship
  we have the following situation: All edges along $\mathbb{P}(u,x)$ are
  labeled $0$; exactly one edge along $\mathbb{P}(u,y)$ is labeled $1$,
  where $u=\lca{x,y}$; all edges along $\mathbb{P}(v,z)$ are labeled $0$,
  and exactly one edge along $\mathbb{P}(v,y)$ is labeled $1$, where
  $v=\lca{y,z}$.  Clearly, $\lca{x,y,z}\in \{u,v\}$.  If $u=v$, then all
  edges in $\mathbb{P}(u,x)$ and $\mathbb{P}(u,z)$ are labeled $0$,
  implying that $x\Ro y$, contradicting that $\Ro$ is discrete.
	
  Now assume that $u=\lca{x,y}\neq v=\lca{y,z}$.  Hence, one of the triples
  $\rt{xy|z}$ or $\rt{yz|x}$ must be displayed by $T(Q)$.  W.l.o.g., we can
  assume that $\rt{yz|x}$ is displayed, since the case $\rt{xy|z}$ is shown
  analogously by interchanging the role of $x$ and $z$.  Thus, $\lca{x,y,z}
  = \lca{x, y} = u\neq \lca{yz}=v$. Hence, $\mathbb{P}(u,y) =
  \mathbb{P}(u,v) \cup \mathbb{P}(v,y)$.  Since $z\Rld y$, the path
  $\mathbb{P}(v,y)$ contains a single 1-edge and $\mathbb{P}(v,z)$ contains
  only 0-edges. Therefore, the paths $\mathbb{P}(u,x)$ and
  $\mathbb{P}(u,v)$ contain only 0-edges, since $x\Rl y$.  Since
  $\mathbb{P}(x,z) = \mathbb{P}(x,u)\cup \mathbb{P}(u,v)\cup
  \mathbb{P}(v,z)$ and all edges along $\mathbb{P}(u,x)$, $\mathbb{P}(u,v)$
  and $\mathbb{P}(v,z)$ are labeled $0$, we obtain $x\Ro z$, again a
  contradiction. 
\end{proof}

Taken together, we obtain the following immediate implication:
\begin{cor}
  The graph $G(\Rld)/\Ro$ does not contain a pair of edges of the form $xv$
  and $yv$.
\label{cor:edge-pointing}
\end{cor}

Recall that the connected components $\underline{Q}$ in $G(\Rld)/\Ro$ are
trees. By Cor.\ \ref{cor:edge-pointing}, $Q$ must be composed of
distinct paths that ``point away'' from each other.  In other words,
  let $P$ and $P'$ be distinct directed path in $Q$ that share a vertex
  $v$, then it is never the case that there is an edge $xv$ in $P$ and an
  edge $yv$ in $P'$, that is, both edges ``pointing'' to the same vertex
  $v$. We first consider directed paths in isolation.

\begin{lemma}
  Let $Q$ be a connected component in $G(\Rld)/\Ro$ that is a
  directed path with $n\ge 3$ vertices labeled $x_1,\dots,x_n$ such that
  $x_ix_{i+1}\in E(Q)$, $1\leq i\leq n-1$. 
  Then the tree $T(Q)$ explaining $Q$ must display all triples in $\mathcal
  R_Q = \{\rt{x_ix_j|x_l} \mid i,j>l\geq 1\}$.  Hence, $T(Q)$ must display
  $\binom{n}{3}$ triples and is therefore the unique (least resolved)
  binary rooted tree $(\dots(x_n,x_{n-1})x_{n-2})\dots)x_2)x_1$ that
  explains $Q$.  Moreover, all interior edges in $T(Q)$ and the edge incident
  to $x_n$ are labeled $1$ while all other edges are labeled $0$.
  \label{lem:path-tree}
\end{lemma}
\begin{proof}
  Let $Q$ be a directed path as specified in the lemma. We prove the
  statement by induction. For $n=3$ the statement follows from Lemma
  \ref{lem:case-i}. Assume the statement is true for $n=k$. Let $Q$ be a
  directed path with vertices $x_1,\dots,x_k, x_{k+1}$ and edges
  $x_ix_{i+1}$, $1\leq i\leq k$ and let $T(Q)$ be a tree that explains
  $Q$. For the subpath $Q'$ on the vertices $x_2,\dots,x_{k+1}$ we can
  apply the induction hypothesis and conclude that $T(Q')$ must display the
  triples $\rt{x_ix_j|x_l}$ with $i,j>l\geq 2$ and that all interior edges in
  $T(Q')$ and the edge incident to $x_{k+1}$ are labeled $1$ while all
  other edges are labeled $0$. Since $T(Q)$ must explain in particular the
  subpath $Q'$ and since $T(Q')$ is fully resolved, we can conclude that
  $T(Q')$ is displayed by $T(Q)$ and that all edges in $T(Q)$ that are also
  in $T(Q')$ retain the same label as in $T(Q')$ .
	
  Thus $T(Q)$ displays in particular the triples $\rt{x_ix_j|x_l}$ with
  $i,j>l\geq 2$.  By Lemma \ref{lem:case-i}, and because there are edges
  $x_1x_2$ and $x_2x_3$, we see that $T(Q)$ must also display
  $\rt{x_2x_3|x_1}$.  Take any triple $\rt{x_3x_j|x_2}$, $j>3$.
  Application of the triple-inference rules shows that any tree that
  displays $\rt{x_2x_3|x_1}$ and $\rt{x_3x_j|x_2}$ must also display
  $\rt{x_2x_j|x_1}$ and $\rt{x_3x_j|x_1}$. Hence, $T(Q)$ must display these
  triples.  Now we apply the same argument to the triples $\rt{x_2x_j|x_1}$
  and $\rt{x_ix_j|x_2}$, $i,j>2$ and conclude that in particular, the
  triple $\rt{x_ix_j|x_1}$ must be displayed by $T(Q)$ and thus, the the
  entire set of triples $\{\rt{x_ix_j|x_l} \colon i,j>l\geq 1\}$.  Hence,
  there are $\binom{n}{3}$ triples and thus, the set of triples that needs
  to be displayed by $T(Q)$ is strictly dense. Making use of a technical
  result from \cite[Suppl. Material]{Hellmuth:15a}, we obtain that $T(Q)$
  is the unique binary tree $(\dots(x_n,x_{n-1})x_{n-2})\dots)x_2)x_1$. Now
  it is an easy exercise to verify that the remaining edge containing $x_1$
  must be labeled $0$, while the interior edge not contained in $T(Q')$ must
  all be 1-edges. 
\end{proof}

\begin{figure}
\begin{center}
\includegraphics[width=1.\textwidth]{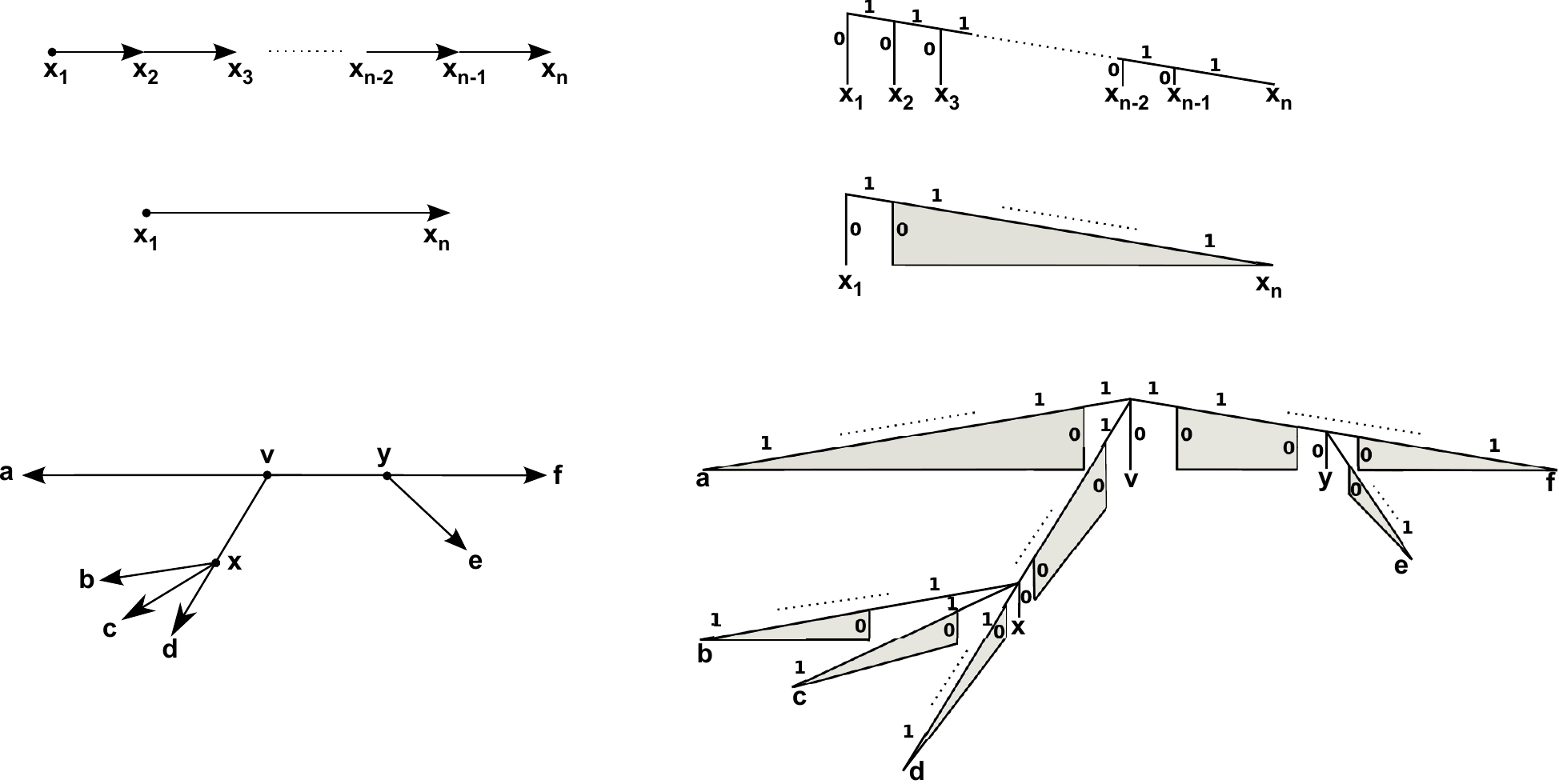}
\end{center}
\caption{Connected components $Q$ in $G(\Rld)/\Ro$ are trees composed of
  paths pointing away from each other.  On the top, a directed path
  $\mathbb P(x_1,x_n)$ together with its unique least resolved tree
  representation is shown. In the middle, an abstract picture of the path
  $\mathbb P(x_1,x_n)$ and its tree is given.  Bottom, a larger example of
  a connected component $Q$ in $G(\Rld)/\Ro$ and its tree-representation
  are sketched.}
\label{fig:Qtree}
\end{figure}

If $Q$ is connected but not a simple path, it is a tree composed of the
paths pointing away from each other as shown in Fig. \ref{fig:Qtree}. It
remains to show how to connect the distinct trees that explain these paths
to obtain a tree $T(Q)$ for $Q$.  To this end, we show first that there is
a unique vertex $v$ in $Q$ such that no edge ends in $v$.

\begin{lemma}
  Let $Q$ be a connected component in $G(\Rld)/\Ro$.  Then there is a
  unique vertex $v$ in $Q$ such that there is no edge $xv\in E(Q)$.
  \label{lem:unique-v}
\end{lemma}
\begin{proof}
  Corollary \ref{cor:edge-pointing} implies that for each vertex $v$ in $Q$
  there is at most one edge $xv\in E(Q)$.  If for all vertices $w$ in $Q$
  we would have an edge $xw\in E(Q)$, then $Q$ contains cycles,
  contradicting the tree structure of $Q$. Hence, there is at least one
  vertex $v$ so that there is no edge of the form $xv\in E(Q)$.  \rev{
    Assume there are two vertices $v,v'$ so that there are no edges of the
    form $xv, yv'$, then all edges incident to $v,v'$ are of the form $vx,
    v'y$. However, in this case, the unique path connecting $v,v'$ in $Q$
    must contain two edges of the form $aw$, $bw$; a contradiction to
    Corollary \ref{cor:edge-pointing}.}  Thus, there is exactly one vertex
  $v$ in $Q$ such that there is no edge $xv\in E(Q)$. 
\end{proof}

By Lemma \ref{lem:unique-v}, for each connected component $Q$ of
$G(\Rld)/\Ro$ there is a unique vertex $v_Q$ s.t.\ all edges incident to
$v_Q$ are of the form $v_Qx$. That is, all directed paths that are maximal
w.r.t.\ inclusion start in $v_Q$.  Let $\mathcal P_Q$ denote the sets of
all such maximal paths.  Thus, for each path $P \in \mathcal P_Q$
there is the triple set $\mathcal{R}_{P}$ according to Lemma
\ref{lem:path-tree} that must be displayed by any tree that explains also
$Q$.  Therefore, $T(Q)$ must display all triples in $\mathcal{R}_Q =
\cup_{P\in \mathcal{P}_Q} \mathcal{R}_{P}$.

The underlying undirected graph $\underline{G(\Rld)/\Ro}$ is isomorphic to
$G(\Rl)/\Ro$. Thus, with Algorithm \ref{alg:Q}, one can similar to the
unrooted case, first construct the tree $T(\underline{Q})$ and then set the
root $\rho_Q=v'_Q$ to obtain $T(Q)$.  It is easy to verify that this tree
$T(Q)$ displays all triples in $\mathcal{R}_Q$.  Moreover, any
edge-contradiction in $T(Q)$ leads to the loss of an input triple
$\mathcal{R}_Q$ and in particular, to a wrong pair of vertices w.r.t.\
$\Rld$ or $\Ro$.  Thus, $T(Q)$ is a least resolved tree for
$\mathcal{R}_Q$ and therefore, a least resolved tree that explains $Q$.

We summarize these arguments in
\begin{cor} 
  Let $Q$ be a connected component in $G(\Rld)/\Ro$. \rev{Then a tree
    $T(Q)$ that explains $Q$ can be obtained from the unique minimally
    resolved tree $T(\underline{Q})$ that explains $\underline{Q}$} by
  choosing the unique vertex $v$ where all edges incident to $v$ are of the
  form $vx$ as the root $\rho_Q$.
\label{cor:Q-root}
\end{cor}

If $G(\Rld)/\Ro$ is disconnected, one can apply Algorithm \ref{alg:all}, to
obtain the tree $T(G(\Rl)/\Ro)$ and then chose either one of the vertices
$\rho_Q$ or the vertex $z_T$ as root to obtain $T(G(\Rld)/\Ro)$, in which
case all triples of $\mathcal{R}_{G(\Rld)/\Ro} = \cup_Q \mathcal R_Q$ are
displayed. Again, it is easy to verify that any edge-contradiction leads to
a wrong pair of vertices in $\Rld$ or $\Ro$. Thus, $T(G(\Rld)/\Ro)$ is a
least resolved tree for $\mathcal{R}_{G(\Rld)/\Ro}$.

To obtain uniqueness of minimally resolved trees 
one can apply similar arguments as in the proofs of
Theorems \ref{thm:connComp} and \ref{thm:star-tree}. This yields the
following characterization: 

\begin{thm}
  Let $Q_1,\dots Q_k$ be the connected components in $G(\Rld)/\Ro$.  Up to
  the choice of the vertices $q'_i$ in Line \ref{item:z} of Alg.\
  \ref{alg:all} for the construction of $T(\underline{Q_i})$ and the choice
  of the root $\rho\in \{\rho_{Q_1},\ldots \rho_{Q_k}, z\}$, the tree $T^*
  = T(G(\Rld)/\Ro))$ is the unique \rev{minimally} resolved tree that
  explains $G(\Rld)/\Ro$.
  \label{thm:star-tree-dir}
\end{thm}

\section{Mix of symmetric and anti-symmetric relations}
\label{sect:mixed}

In real data, e.g., in the application to mitochondrial genome
arrangements, one can expect that the known relationships are in part
directed and in part undirected. Such data are naturally encoded by a
relation $\Rld$ with directional information and a relation $\Rlstar$
comprising the set of pairs for which it is unknown whether \rev{ one of
  $x\Rld y$ and $y\Rld x$ or $x\Rl y$ are true. Here, $\Rl$ is a subset of
  $\Rlstar$.}  The disjoint union $\Rlstar\uplus\Rld$ of these two parts
can be seen as refinement of a corresponding symmetrized relation $x\Rl y$.
Ignoring the directional information one can still construct the tree
$T(G(\Rl)/\Ro)$.  In general there will be less information of the
placement of the root in $T(G(\Rlstar\cup\Rld)/\Ro)$ than with a fully
directed edge set.

\rev{In what follows, we will consider all edges of
  $G(\Rlstar\cup\Rld)/\Ro$ to be directed, that is, for a symmetric pair
  $(a,b)\in\Rlstar$ we assume that both directed edges $ab$ and $ba$ are
  contained in $G(\Rlstar\cup\Rld)/\Ro$.  Still, for any connected
  component $Q$ the underlying undirected graph $\underline{Q}$ is a tree.}
Given a component $Q$ we say that a directed edge $xy \in E(Q)$
\emph{points away from the vertex $v$} if the unique path in
$\underline{Q}$ from $v$ to $x$ does not contain $y$. In this case the path
from $v$ to $y$ must contain $x$.  Note that in this way we defined
``pointing away from $v$'' not only for the edges incident to $v$, but for
all directed edges.  A vertex $v$ is a \emph{central vertex}
if, for any two distinct vertices $x,y\in V$ that form an edge in $T$,
either $xy$ or $yx$ in $T$ points away from $v$.

As an example consider the tree $a\leftarrow b\rightarrow c \leftrightarrow
d\rightarrow e$. \rev{There is only the edge $bc$ containing $b$ and $c$.
However, $bc$ does not point away from vertex $d$, since the unique path
from $d$ to $b$ contains $c$.}  
Thus $d$ is not central. On the other hand, $b$ is a
central vertex. The only possibility in this example to obtain a valid
relation $\Rld$ that can be displayed by rooted 0/1-edge-labeled tree is
provided by removing the edge $dc$, since otherwise Cor.\
\ref{cor:edge-pointing} would be violated.

In the following, for given relations $\Rlstar$ and $\Rld$ we will denote
with $\Rldstar$ a relation that contains $\Rld$ and exactly one pair,
either $(x,y)$ or $(y,x)$, from $\Rlstar$.

\begin{lemma}
  For a given graph $G(\Rlstar\cup\Rld)/\Ro$ the following statements
  are equivalent:
  \begin{itemize}
  \item[(i)] There is a relation $\Rldstar$ that is the antisymmetric
    single-1-relation of some 0/1-edge-labeled tree.
  \item[(ii)] There is a central vertex in each connected component $Q$ of
    $G(\Rlstar\cup\Rld)/\Ro$.
  \end{itemize}
\end{lemma}
\begin{proof}
  If there is a relation $\Rldstar$ that can be displayed by a rooted
  0/1-edge-labeled tree, then $G(\Rldstar)/\Ro$ consists of connected
  components $Q$ where each connected component is a tree composed of
  maximal directed paths that point away from each other.  Hence, for each
  connected component $Q$ there is the unique vertex $v_Q$ such that all
  edges incident to $v_Q$ are of the form $v_Qx$ and, in particular, $v_Q$
  is a central vertex $v_Q$ in $Q$ and thus, in
  $G(\Rlstar\cup\Rld)/\Ro$.

  Conversely, assume that each connected component $Q$ has a central vertex
  $v_Q$. Hence, one can remove all edges that do not point away
  from $v_Q$ and hence obtain a connected component $Q'$ that is
  still a tree with $V(Q)=V(Q')$ so that all maximal directed paths point
  away from each other and in particular, start in $v_Q$. Thus,
  for the central vertex $v_Q$ all edges incident to
  $v_Q$ are of the form $v_Qx$. Since $Q'$ is now a
  connected component in $G(\Rldstar)/\Ro$, we can apply Cor.\
  \ref{cor:Q-root} to obtain the tree $T(Q')$ and Thm.\
  \ref{thm:star-tree-dir} to obtain $T(G(\Rldstar)/\Ro)$.  
\end{proof}

The key consequence of this result is the following characterization of the
constraints on the possible placements of the root.

\begin{cor}
  Let $Q$ be a connected component in $G(\Rlstar\cup\Rld)/\Ro$ and let
  $T(\underline{Q})$ be the unique least resolved tree that explains the
  underlying undirected graph $\underline{Q}$.  Then each copy $v'$ of a
  vertex $v$ in $Q$ can be chosen to be the root in $T(\underline{Q})$ to
  obtain $T(Q)$ if and only if $v$ is a central vertex in $Q$.
\end{cor}

\section{Concluding Remarks} 

In this contribution we have introduced a class of binary relations
deriving in a natural way from edge-labeled trees. This construction has
been inspired by the conceptually similar class of relations induced by
vertex-labeled trees \cite{HW:17}.  The latter have co-graph structure and
are closely related to orthology and paralogy
\cite{Hellmuth:13a,Lafond:14,Hellmuth:15a,HW:15}. Defining $x\sim y$
whenever at least one 1-edge lies along the path from $x$ to $y$ is related
to the notion of xenology: the edges labeled 1 correspond to horizontal
gene transfer events, while the 0-edge encode vertical transmission. In its
simplest setting, this idea can also be combined with vertex labels,
leading to the directed analog of co-graphs \cite{HSW:16}. Here, we have
explored an even simpler special case: the existence of a single 1-label
along the connecting path, which captures the structure of rare event data
as we have discussed in the introduction.  We have succeeded here in giving
a complete characterization of the relationships between admissible
relations, which turned out to be forests, and the underlying phylogenetic
tree.  Moreover, for all such cases we gave polynomial-time algorithms to
compute the trees explaining the respective relation.

\rev{The characterization of single-1 relations is of immediate relevance
  for the use of rare events in molecular phylogenetics. In particular, it
  determines lower bounds on the required data: if too few events are
  known, many taxa remain in $\Ro$ relation and thus unresolved. On the
  other hand, if taxa are spread too unevenly, $G(\Rld)/\Ro$ will be
  disconnected, consisting of connected components separated by multiple
  events. The approach discussed here, of course, is of practical use only
  if the available event data are very sparse. If taxa are typically
  separated by multiple events, classical phylogenetic approaches, i.e.,
  maximum parsimony, maximum likelihood, or Bayesian methods, will
  certainly be preferable. An advantage of using the single-1 relation is
  that it does not require the product structure of independent characters
  implicit in the usual, sequence or character-based methods, nor does it
  require any knowledge of the algebraic properties of the underlying
  operations as, e.g., in phylogenetic reconstruction from (mitochondrial)
  gene order rearrangements. This begs the question under which conditions
  on the input data identical results are obtained from the direct
  translation of the single-1 relation and maximum parsimony on character
  data or break point methods for genome rearrangments. 
}

\rev{A potentially useful practical application of our results is a 	new
  strategy to incorporate rare event data into conventional phylogenetic
  approaches. Trees obtained from a single-1 relation will in general be
  poorly resolved. Nevertheless, they determine some monophyletic groups
  (in the directed case) or a set of split (in the undirected case) that
  have to be present in the true phylogeny. Many of the software tools
  commonly used in molecular phylogenetic can incorporate this type of
  constraints. Assuming the there is good evidence that the rare events are
  homoplasy-free, the $T(G(\Rld)/\Ro)$ represents the complete information
  from the rare events that can be used constrain the tree reconstruction
  process.}

The analysis presented here makes extensive use of the particular
properties of the single-1 relation and hence does not seem to generalize
easily to other interesting cases. Horizontal gene transfer, for example,
is expressed naturally in terms of the ``at-least-one-1'' relation
$\Rd$. It is worth noting that $\Rd$ also has properties (L1) and (L2) and
hence behaves well w.r.t.\ contraction of the underlying tree and
restriction to subsets of leaves. Whether this is sufficient to obtain a
complete characterization remains an open question.

Several general questions arise naturally. For instance, is there a
characterization of admissible relations in terms of forbidden subgraphs
graphs or minors? For instance, the relation $\Rld/\Ro$ is characterized in
terms of the forbidden subgraph $x\rightarrow v \leftarrow y$. Hence, it
would be of interest, whether such characterizations can be derived for
arbitrary relations $\Rldk$ or for $\Rd$.  If so, can these forbidden
substructures be inferred in a rational manner from properties of vertex
and/or edge labels along the connecting paths in the explaining tree?
Is this the case at least for labels and predicates satisfying (L1)
  and (L2)?

\section*{Acknowledgements}
  \rev{We thank the anonymous reviewers for the example in Fig.\
    \ref{fig:nonU-LRT} and numerous valuable comments that helped us to
    streamline the presentation and to shorten the proofs.}  This work was
  funded by the German Research Foundation (DFG) (Proj.\ No.\ MI439/14-1 to
  PFS).  YJL acknowledges support of National Natural Science Foundation
  of China (No. 11671258) and Postdoctoral Science Foundation of China
  (No.\ 2016M601576)

\bibliography{maribelismo}


\end{document}